\documentclass[12pt,a4paper,aip,jmp]{revtex4-1}

\usepackage{graphics}
\usepackage{graphicx}
\usepackage{amsmath}
\usepackage{amsthm}
\usepackage{amssymb}
\usepackage{graphpap}
\usepackage[arrow,curve,matrix]{xy}

\newtheorem{theorem}{Theorem}
\newtheorem{lemma}[theorem]{Lemma}

\def\defn#1{{\em #1}}

\def\Real{{\mathbb R}}
\def\Cmpx{{\mathbb C}}

\def\cnj#1{{{#1}^*}}
\def\dual#1{{\widetilde{#1}}}
\def\union{\cup}
\def\Union{\bigcup}
\def\inter{\cap}

\def\innerprod(#1,#2){{\left<#1\,,\,#2\right>}}
\def\Set#1{{\left\{#1\right\}}}
\def\Re{{\textup{Re}}}
\def\Im{{\textup{Im}}}

\def\qquadtext#1{{\qquad\text{#1}\qquad}}
\def\qquadand{\qquadtext{and}}
\def\quadtext#1{{\quad\text{#1}\quad}}
\def\quadand{\quadtext{and}}

\def\pfrac#1#2{{\frac{\partial #1}{\partial #2}}}

\def\Ebun{{\cal E}}
\def\Nbun{{\cal N}}
\def\Horiz{{\cal H}}
\def\Vert{{\cal V}}
\def\Cd{{\dot C}}
\def\Adot{\dot A}
\def\supp{\textup{supp}}

\def\imu{{\mu}}
\def\inu{{\nu}}
\def\isigma{{\sigma}}

\def\Pihom{\Pi_{\textup{\bf H}}}
\def\Zhom{Z_{\textup{\bf H}}}
\def\What{\hat W}
\def\IC#1#2{{i^{(#1)}_{#2}}}
\def\len#1{\left|#1\right|}

\def\MX{{M_X}}
\def\MY{{M_Y}}
\def\EbunX{{\Ebun_X}}

\def\PiX{p_X}
\def\PiY{p_Y}
\def\piX{{\pi_X}}
\def\piY{{\pi_Y}}
\def\ppY{\varpi_Y}
\def\ppX{\varpi_X}
\def\psihat{{\hat\psi}}

\def\detg{|\det g|}

\def\DomPsi{{\cal D}}
\def\hashx{\star_{\!X}}

\def\MF{{\boldsymbol F}}
\def\MD^#1_#2{{\boldsymbol D^\speciesa}^{#1}{}_{#2}}
\def\ML^#1_#2{{\boldsymbol L^\speciesa}^{#1}{}_{#2}}
\def\ID_#1^#2{{\boldsymbol D^\speciesa}_{#1}{}^{#2}}
\def\IL_#1^#2{{\boldsymbol L^\speciesa}_{#1}{}^{#2}}
\def\RL{{\boldsymbol L^\speciesa}}

\def\speciesa{{\text{\tiny$\lfloor\!\alpha\!\rceil$}}}
\def\speciesas#1{{\text{\tiny$\lfloor\!\alpha_{#1}\!\rceil$}}}
\def\speciesabig{{\text{$\lfloor\!\alpha\!\rceil$}}}
\def\speciesel{{\text{$\lfloor\!\textup{el}\!\rceil$}}}
\def\speciesion{{\text{$\lfloor\!\textup{ion}\!\rceil$}}}

\def\chione{\chi^{\textup{\bf H}}}
\def\chione{X}
\def\FT{\hat}
\def\FTPiHom{\FT{\Pi}_{\textup{\bf H}}}
\def\kHat{\hat{k}}
\def\hHat{\hat{h}}
\def\omegaHat{\hat{\omega}}
\def\kkHat{\hat{k}}
\def\kk{k}
\def\II{I}

\def\mapint{\lefteqn{\text{\scriptsize$\ \,\triangle$}}\int}
\def\mapints{\lefteqn{\text{\tiny$\ \triangle$}}\int}
\def\Vsig{\underline\sigma}
\def\zetabd{{\boldsymbol\zeta}}

\def\zetao#1{{\zeta^\speciesa_{1,#1}}}
\def\zetaoh#1{{\check\zeta^\speciesa_{1,#1}}}

\def\SigmaM{\Sigma_M}
\def\SigmaMX{\Sigma_{M_X}}
\def\SigmaMY{\Sigma_{M_Y}}
\def\SigmaE{\Sigma_{\Ebun}}
\def\SigmaEX{\Sigma_{\Ebun_X}}
\def\SigmaEY{\Sigma_{\Ebun_Y}}
\def\MXp{{M_X^+}}
\def\MYp{{M_Y^+}}
\def\EbunXp{{\Ebun_X^+}}
\def\EbunYp{{\Ebun_Y^+}}
\def\Mp{{M^+}}
\def\Ebunp{{\Ebun^+}}
\def\Zonez{\check Z_1}

\def\SigmaN{\Sigma_N}

\def\kappaQM{{\cal Q}_0^2}

\begin{document}

\title{Covariant Constitutive Relations and Relativistic Inhomogeneous Plasmas}
\author{J Gratus}
\author{R W Tucker}
\affiliation{Physics Department Lancaster University and the Cockcroft
Institute}

\begin{abstract}
The notion of a two-point {\it susceptibility kernel} used to describe
linear electromagnetic responses of dispersive continuous media in
non-relativistic phenomena is generalized to accommodate the
constraints required of a causal formulation in spacetimes with
background gravitational fields.  In particular the concepts of
spatial material inhomogeneity and temporal non-stationarity are
formulated within a fully covariant spacetime framework.  This
framework is illustrated by re-casting the Maxwell-Vlasov equations
for a collisionless plasma in a form that exposes a 2-point
electromagnetic susceptibility kernel in spacetime.  This permits the
establishment of a perturbative scheme for non-stationary
inhomogeneous plasma configurations. Explicit formulae for the
perturbed kernel are derived in both the presence and absence of
gravitation using the general solution to the relativistic equations
of motion of the plasma constituents. In the absence of gravitation
this permits an analysis of collisionless damping in terms of a system
of integral equations that reduce to standard Landau damping of
Langmuir modes when the perturbation refers to a homogeneous
stationary plasma configuration.  It is concluded that constitutive
modelling in terms of a 2-point susceptibility kernel in a covariant
spacetime framework offers a natural extension of standard
non-relativistic descriptions of simple media and that its use for
describing linear responses of more general dispersive media has wide
applicability in relativistic plasma modelling.
\end{abstract}


\maketitle

\section{Introduction}
\label{ch_Disp}

The behaviour of a material medium in response to electromagnetic and
gravitational fields encompasses a vast range of classical and quantum
physics. For media composed of a large collection of molecular or
ionized structures recourse to a statistical description is required
and this often leads to a coarser description in terms of a few
thermodynamic variables and their correlations. Such a description
relies on the efficacy of particular constitutive models or
phenomenological constitutive data that serve to circumscribe its
domain of applicability.

For phenomena where the relative motions of the constituents approach
the speed of light in vacuo or the material experiences bulk
accelerations or gravitational interactions such constitutive
descriptions must be formulated within a relativistic framework.
However even within a spacetime covariant formulation there remains
great freedom in how to accommodate electromagnetic responses that depend on
material dispersion induced by spatial correlations or temporal delays
of electromagnetic interactions\cite{gratus2010covariant}. The incorporation of such effects in a
theoretical description often relies on a detailed structural model of
the medium particularly if it is inhomogeneous or external
gravitational gradients are relevant. Notwithstanding these
complexities simple constitutive models have proved of considerable
value for homogeneous polarizable media that exhibit
temporal dispersion in a laboratory frame where gravity plays no
essential role. Indeed the notion of permittivity and permeability
tensors is often adequate to parametrize a large range of
experimental linear responses of simple polarizable media to external
static and dynamic electromagnetic fields.  More generally, for non-dispersive
media these tensors can be subsumed into a {\it susceptibility kernel}
that readily accommodates special relativistic effects on the bulk
motion of media.

In this article the degree to which the notion of a {\it
  susceptibility kernel} can be generalized to describe linear
electromagnetic responses of dispersive continuous media is
explored. In particular the effects of spatial material inhomogeneity
and non-stationarity will be formulated within a fully covariant
spacetime framework. In this manner the formulation can accommodate
arbitrary gravitational and electromagnetic interactions. The
framework will be illustrated by re-casting the Maxwell-Vlasov
equations for a collisionless plasma in a form that exposes a
2-point\footnote{Points here refer in general to events in a spacetime
  manifold.}  electromagnetic susceptibility kernel in an arbitrary
external gravitational field. This permits the establishment of a
perturbative scheme for non-stationary inhomogeneous plasma
configurations in terms of such a kernel. Explicit formulae for the
perturbed kernel are derived in both the presence and absence of
gravitation in terms of the general solution to the equations of
motion of the plasma constituents. In the absence of gravitation this
permits an analysis of collisionless damping in terms of a system of
integral equations that reduce to standard Landau damping of Langmuir
modes when the perturbation refers to a homogeneous stationary plasma
configuration.

It is concluded that constitutive modelling in terms of a 2-point
susceptibility kernel in a covariant spacetime framework offers a
natural extension of standard non-relativistic descriptions of simple
media and that its use for describing linear responses of more general
dispersive media has wide applicability in relativistic plasma
modelling.

\section{Constitutive Relations}

\label{ch_Notation}

In the following spacetime $M$ is considered a globally hyperbolic,
topologically trivial four dimensional manifold endowed with a metric
tensor $g$ with signature $(-1,+1,+1,+1)$ describing gravitation.  A
closed 2-form $F$ describes the electromagnetic field. The bundle of
exterior $p-$forms over $M$ is denoted $\Lambda^p M$ and its sections
$\Gamma\Lambda^p M$ are $p-$forms on $M$. The bundle of all forms is
$\Lambda M= \Union_{p=0}^{p=4} \Lambda^p M$.  Associated with $g$ is
the Hodge map $\star$. Thus for $\alpha\in\Gamma\Lambda^p M$ its
corresponding Hodge dual is denoted $\star\alpha \in \Gamma^{4-p}
\Lambda M$.  The tangent bundle over $M$ is denoted $TM$ and its
sections $\Gamma TM$ are vector fields on $M$.  We call the 1-form
$\dual{J}=g(J,-)\in\Gamma\Lambda^1 M$ the {\it metric dual} of the
vector field $J\in\Gamma TM$.  Maxwell's equations for the
electromagnetic field $F\in\Gamma\Lambda^2 M$ in a polarizable medium
containing an electric current $J\in\Gamma TM$, satisfying the
continuity (or current conservation) equation
$d\star\dual{J}=0$, are written
\begin{align}
d F=0 \qquadand d\star G = - \star \dual{J}
\label{Disp_Maxwell}
\end{align}
The excitation 2-form $G\in\Gamma\Lambda^2 M$ can always be
expressed
\begin{align}
G = \epsilon_0 F + \Pi
\label{Disp_def_G}
\end{align}
in terms of the permittivity $\epsilon_0$ of free space.  The
polarization\footnote{In this article the term polarization will refer
  to any state of the medium that gives rise to magnetization or
  electrical polarization in some frame} 2-form
$\Pi\in\Gamma\Lambda^2 M$ results from all electromagnetic
field sources not made explicit in $J$.

In general $\Pi$ and $J$ are non-linear functionals of $F$
and other fields such as matter and initial data on any initial spacelike
hypersurface $\SigmaM\subset M$.  Such functionals are the {\it{constitutive
    relations}} describing $G$ and $J$ in terms of $F$ and these other
fields.

It is convenient to introduce integration on a fibred manifold $\Nbun$
of dimension $n+r$ with projection $\pi_\Nbun:\Nbun\to N$ over a
manifold $N$ of dimension $n$. Thus at each point $\sigma\in N$ one
has the fibre $\Nbun_\sigma=\pi_\Nbun^{-1}\Set{\sigma}=
\Set{(\sigma',\varsigma)\in\Nbun\,\big|\,
  \pi_\Nbun(\sigma',\varsigma)=\sigma}$ so $\dim(\Nbun_\sigma)=r$ is
the fibre dimension.  For $\alpha\in\Gamma\Lambda^{p+r}\Nbun$ we
define\cite{bott1982differential,de1984differentiable} the form
$\mapints_{\pi_\Nbun}\alpha\in\Gamma\Lambda^{p} N$ by
\begin{align}
\int_N\beta\wedge\mapint_{\pi_\Nbun}\alpha
=
\int_\Nbun \pi_\Nbun^\star(\beta)\wedge \alpha
\label{Notation_mapint}
\end{align}
for all $\beta\in\Gamma\Lambda^{n-p} N$.

In terms of local coordinates $(\sigma^1,\ldots,\sigma^n)$ and
$(\sigma^1,\ldots,\sigma^n,\varsigma^1\ldots \varsigma^r)$ for patches
on $N$ and $\Nbun$
respectively, one may write the fibre integral
\begin{equation}
\begin{aligned}
\bigg(\mapint_{\pi_\Nbun}\alpha\bigg)\bigg|_\sigma
&=
\sum_{1\le I_1<\ldots < I_p\le n}
d\sigma^{I_1}\wedge\ldots\wedge d\sigma^{I_p}
\int_{\varsigma\in \Nbun_\sigma}
i_{\partial/\partial \sigma^{I_p}} \ldots i_{\partial/\partial \sigma^{I_1}}
\alpha|_{(\sigma,\varsigma)}
\end{aligned}
\label{Disp_def_fibre_int}
\end{equation}
where $\Nbun_\sigma=\pi_\Nbun^{-1}(\Set{\sigma})$ is the fibre over the point
$\sigma\in N$ and $i_{\partial/\partial \sigma^{I_k}}$ is the contraction
on forms.  Observe that if $\alpha$ does not contain the factor
$d\varsigma^1\wedge\cdots\wedge d\varsigma^r$ then $\mapints_{\pi_\Nbun
}\alpha=0$. The proof of this is given in appendix lemma \ref{lm_fibre_int}.

A key result of fibre integration, used to establish the current
continuity equation, is that it commutes with the exterior derivative:
\begin{align}
\bigg(d \mapint_{\pi_\Nbun} \alpha\bigg)\bigg|_{\sigma}
= \bigg(\mapint_{\pi_\Nbun} d \alpha\bigg)\bigg|_{\sigma}
\label{Notation_mapint_comm_d}
\end{align}
for $\sigma$ not on the boundary of $N$ provided the support of $\alpha$
does not intersect the boundary of $\Nbun$. The proof is given in
appendix lemma \ref{lm_d_mapint}.

In general models for $\Pi$ demand a knowledge of the dynamics of
sources responsible for polarization as well as any permanent
polarization that may exist in the medium. A full dynamical description
depends on a specification of appropriate initial value data
$\zetabd$ on $\SigmaM$. The exact structure of $\zetabd$ depends on
the sources of the polarization. For the plasma model described in
section \ref{ch_Plasma} the initial data corresponds to the velocity
profile for each particle species at each point on $\SigmaM$ in the
plasma.

In this article $\Pi$ is considered to be an {\it affine}
functional of $F$ of the form
\begin{align}
\Pi[F,\zetabd]= \mapint_{\PiX} \chi \wedge \PiY^\star(F) + Z[\zetabd]
\label{Disp_Pi}
\end{align}
for some functional $Z$ of $\zetabd$.
The first term on the right is  expressed in terms of the
fibre integral of a two-point \defn{susceptibility kernel}
$\chi\in\Gamma\Lambda^4(\MX\times\MY)$ expressible locally as
\begin{align}
\chi=\tfrac14\chi_{abcd}(x,y) dx^a\wedge dx^b\wedge dy^c\wedge dy^d
\label{Disp_chi_coords}
\end{align}
Here $\MX$ and $\MY$ are two copies of $M$, locally coordinated by
$(x^0,\ldots,x^3)$ and $(y^0,\ldots,y^3)$ respectively, with
projections $\PiX:\MX\times\MY\to\MX$, $\PiY:\MX\times\MY\to\MY$,
$\PiX(x,y)=x$, $\PiY(x,y)=y$ and initial hypersurfaces
$\SigmaMX\subset\MX$ and $\SigmaMY\subset\MY$.  Throughout,
summation is over Roman indices $a,b,c=0,1,2,3$ and Greek
indices $\imu,\inu,\isigma=1,2,3$.

To consistently remove any reference to $M$ (without a subscript) let
$F\in\Gamma\Lambda^2\MY$, $\epsilon_0 F\in\Gamma\Lambda^2\MX$,
$G\in\Gamma\Lambda^2\MX$, $J\in\Gamma T\MX$ and
$\Pi[F,\zetabd]\in\Gamma\Lambda^2\MX$.  Thus $\epsilon_0$ can be
regarded as a map $\epsilon_0:\Gamma\Lambda^2\MY\to\Gamma\Lambda^2\MX$
which is the pullback of the natural isomorphism $\MX\to\MY$, together
with a scaling to accommodate the choice of electromagnetic units.

In terms of local coordinate bases on $\MX$ and $\MY$ the components
of (\ref{Disp_Pi}) are
\begin{align}
\Pi[F,\zetabd]_{ab}(x)= \int_{y\in M}\tfrac14 \chi_{abcd}(x,y)\, F_{ef}(y)
\,dy^{cdef}
+ Z[\zetabd]_{ab}
\label{Disp_coords_Pi}
\end{align}
in a multi-index notation with
\begin{align*}
dx^{a_1\ldots a_p}\equiv dx^{a_1}\wedge \cdots \wedge dx^{a_p}
\end{align*}
and
\begin{align*}
\IC{x}{a_1\ldots a_p}\equiv i_{\pfrac{}{x^{a_p}}}\cdots
i_{\pfrac{}{x^{a_1}}}
\end{align*}
(Note the reverse order for internal contraction.)
Summations over multi-indices $I\subset\Set{1,\ldots,n}$ considered as
an ordered $p$-list $I_1<I_2<\ldots<I_p$ of length $\len{I}=p$ will
also be employed. Thus
\begin{align*}
dx^I\equiv dx^{I_1\cdots I_p}= dx^{I_1}\wedge\cdots\wedge dx^{I_p}
\end{align*}
and
\begin{align*}
\IC{x}{I}\equiv \IC{x}{I_1\cdots I_p} = i_{\pfrac{}{x^{I_p}}}\cdots
i_{\pfrac{}{x^{I_1}}}
\end{align*}
so that, via summation, if $\alpha\in\Gamma\Lambda^pM$ then
$dx^I \wedge\IC{x}{I}\alpha = \alpha$ where $\len{I}=p$.

In this notation the product manifold $\MX\times\MY$ inherits the
following maps that will be employed below:
\begin{align*}
& d_X:\Gamma\Lambda^p(\MX\times\MY) \to \Gamma\Lambda^{p+1}
(\MX\times\MY) \,,\quad
\\&\qquad d_X (\alpha) =
  \pfrac{\alpha_{IJ}}{x^a} dx^a\wedge dx^I\wedge dy^J
\\
& d_Y:\Gamma\Lambda^p(\MX\times\MY) \to \Gamma\Lambda^{p+1}
(\MX\times\MY) \,,\quad
\\&\qquad d_Y (\alpha) =
  \pfrac{\alpha_{IJ}}{y^a} dy^a\wedge dx^I\wedge dy^J
\\
& \hashx:\Gamma\Lambda(\MX\times\MY) \to \Gamma\Lambda
(\MX\times\MY) \,,\quad
\\&\qquad \hashx(\alpha) =
  \alpha_{IJ} (\star dx^I)\wedge dy^J
\end{align*}
where $\alpha=\alpha_{IJ}\, dx^I\wedge dy^J$

Since $F=dA$ and for $A$ with compact support away from any boundary of
$\MY$ it follows from (\ref{Disp_Pi}) that
\begin{align*}
\Pi[F,\zetabd]= -\mapint_{\PiX} (d_Y\chi) \wedge \PiY^\star(A) + Z[\zetabd]
\end{align*}
Hence $\Pi[F,\zetabd]$ remains invariant\footnote{When $A$ is not
  compact on $\MY$ invariance is modulo a boundary term.} under the gauge transformation
\begin{align}
\chi\quad{\longrightarrow}\quad \chi + d_Y \check\zeta
\label{Disp_gauge_dYchi}
\end{align}
for any $\check\zeta=\check\zeta_{abc} dx^{ab}\wedge
dy^c\in\Gamma\Lambda^3(\MX\times\MY)$. Since the support of $A$ can be
made arbitrarily small $d_Y\chi$ is uniquely specified by
$\Pi[F,\zetabd]$. Furthermore 
\begin{align*}
d\star \Pi[F,\zetabd]= 
-\mapint_{\PiX} (d_X\star_X d_Y\chi) \wedge \PiY^\star(A) + d\star Z[\zetabd]
\end{align*}
hence $d\star \Pi[F,\zetabd]$ is invariant under the gauge
transformation 
\begin{align}
\chi\quad{\longrightarrow}\quad \chi + d_Y \check\zeta + \hashx d_X\check\xi
\label{Disp_gauge_dstardYchi}
\end{align}
for any $\check\zeta=\check\zeta_{abc} dx^{ab}\wedge dy^c$ and
$\check\xi=\check\xi_{abc} dx^{a}\wedge dy^{bc}$. Similarly $d_X\hashx
d_Y \chi$ is uniquely determined by $d\star \Pi[F,\zetabd]$.

In general, the permittivity functional $\Pi$ is a non-local
functional in spacetime given by the integral
(\ref{Disp_coords_Pi}). If $\chi$ is smooth, and not identically zero,
then $\Pi$ is always non-local. However for distributional
susceptibility kernels it is possible for $\Pi$ to remain local.  In
this category one has the local, linear Minkowski constitutive relations
\begin{align*}
\Pi[F] = \epsilon_0(\epsilon_r-1) i_v F\wedge\dual{v}
+ \epsilon_0(\mu^{-1}_r - 1) \star\big((i_v\star F)\wedge F\big)
\end{align*}
where $v\in\Gamma T \MY$ is a vector field representing the bulk
4-velocity of the medium and $\epsilon_r,\mu_r\in\Gamma\Lambda^0\MY$
are the relative permittivity and permeability scalars of the medium.
These relations can be represented by a distributional susceptibility
kernel with support on the diagonal set
$\Set{(x,y)\in\MX\times\MY|x=y}$.

In general $\Pi$ is said to be causal on all of $M$ if $\Pi|_x$ only
depends of the values of $F$ which lie on or within the past
light-cone\cite{wald1984general}${}^{,}$ \footnote{We write $y\in
  J^-(x)$ if $x$ is (timelike or lightlike) causally
  connected to $y$ and $x$ lies in the future of $y$}
$J^-(x)\subset\MY$ of $x$. If $\Pi$ depends on $\zetabd$ it may be
causal on $\MXp$ where $\MXp=\SigmaMX\union\Set{x\text{ lies to the
    future of }\SigmaMX}$.  The functional $\Pi$ is causal on $\MXp$
if $\Pi[F,\zetabd]|_x$ only depends on the values of $F$ and $\zetabd$
which lie on or within its past light-cone $J^-(x)\inter \MXp$ of $x$
and $x\in\MXp$. The data functional $Z$ is casual on $\MXp$ if
$Z[\zetabd]|_x$ depends only on $\zetabd\in\SigmaMX\inter J^-(x)$ for
all $x\in\MXp$.  For $\Pi$ to be causal on $\MXp$ it is necessary and
sufficient (lemma \ref{lm_Causal} in the appendix) that the following
be satisfied:
\begin{itemize}
\item
$Z$ is causal on $\MXp$,
\item
$(d_Y\chi)|_{(x,y)}=0$
for all $(x,y)\in\MXp\times\MYp$ such that $y \notin J^-(x)$ and
\item
$\iota_{\SigmaMY}^\star(\chi)\,\big|\,_{(x,y)}=0$ for all
$(x,y)\in\MXp\times\SigmaMY$ such that $y \notin J^-(x)$, where
$\iota_{\SigmaMY}:\MXp\times\SigmaMY\hookrightarrow\MXp\times\MYp$
is the natural embedding.
\end{itemize}

\subsection{Spacetime homogeneous constitutive relations for media in
  Minkowski spacetime}
\label{sch_Homo}

Minkowski spacetime has properties that underpin the notions of
material spatial homogeneity and stationary
processes.  Being isomorphic to a real 4-dimensional vector space it
can be given an affine structure in addition to its light-cone
structure.  Physically this implies that no particular point in a
spacetime without gravitation has a distinguished status and the
concepts of material and field energy, momentum and angular momentum
can be defined in terms of the Killing symmetries of the spacetime
metric.  Since all points of the spacetime are equivalent relative to
this affine structure it is sufficient to denote $\MX$ and $\MY$ by
$M$ and, relative to any point chosen as origin, a point with
coordinates $x$ can be identified with a vector denoted by $x\in
\Real^4$.  It is then convenient to introduce the Minkowski
translation map $A_z:M\to M$, $A_z(x)=x+z$ that maps points $x$ to
$x+z$ on $M$.

If the electromagnetic properties of an unbounded medium are independent of
location in spacetime they will be called {\it spacetime
  homogeneous}. Such electromagnetic constitutive properties imply
that variations in $F$ at event $y\in M$ produce an induced variation
in a functional $\Pihom[F]$ at event $x\in M$, via a kernel
$\chi_{abcd}(x,y)$ that depends on the 4-vector $x-y$.  If the
constitutive relation is causal then there is no induced variation if
$x\notin J^+(y)$. Furthermore in a spacetime homogeneous medium
$Z[\zetabd]=\Zhom$ where $\Zhom\in\Gamma\Lambda^2 M$ is independent of
$\zetabd$.

In terms of $A_z$ an electromagnetic constitutive functional $\Pihom$ is
given by
\begin{align}
\Pihom[F]= \mapint_{\PiX} \chi \wedge \PiY^\star(F) + \Zhom
\label{Disp_homo_Pi}
\end{align}
The functional $\Pihom$ is said to be {\it{spacetime
    homogeneous}}\footnote{Note that this definition of homogeneity
  refers only to the electromagnetic properties of a medium.}  if
\begin{align}
\Pihom[A_z^\star F]=A_z^\star \Pihom[F]
\label{Disp_homo_def}
\end{align}
This follows if
the susceptibility
kernel $\chi$ satisfies
\begin{align}
\chi|_{(x+z,y+z)}=\chi|_{(x,y)}
\label{Disp_homo_chi}
\end{align}
and $A_z^\star \Zhom= \Zhom$. The
contribution $\Zhom$ may model the presence of an externally
prescribed stationary uniform permanent magnetic or electric polarization.
Equation (\ref{Disp_homo_chi}) implies the components of $\chi$ in
(\ref{Disp_chi_coords}) can be written
\begin{align}
\chi_{abcd}(x,y)=\chione_{abcd}(x-y)
\label{Disp_homo_hat_chi}
\end{align}
where
\begin{align}
\chione_{abcd}(x) = \chi_{abcd}{(x,0)}
\label{Disp_homo_hat_chi_res}
\end{align}
Thus, in a Minkowski spacetime for materials with electromagnetic spacetime homogeneous
properties, (\ref{Disp_coords_Pi}) can be written in terms of a
convolution integral:
\begin{equation}
\begin{aligned}
\Pihom[F]_{ab}(x)&=\tfrac14 \int_{y\in M}\chione_{abcd}(x-y) F_{ef}(y)
dy^{cdef}
+(\Zhom)_{ab}
\\&\equiv \tfrac14 \epsilon^{cdef}(\chione_{abcd} * F_{ef})(x)
+(\Zhom)_{ab}
\end{aligned}
\label{Disp_homo_conv}
\end{equation}
where $\epsilon^{cdef}= \pm 1,0$ denotes the Levi-Civita alternating
symbol in coordinates in which the metric tensor takes the form
$g=\eta_{ab}dx^a\otimes dx^b$ where
$\eta_{ab}=\text{diag}(-1,+1,+1,+1)$. In these coordinates
the $(\Zhom)_{ab}$ are all constants.

Let $\FT{F}_{ef}(k)$ and $\FTPiHom[F]_{ab}(k)$ denote the Fourier
transforms of $F_{ef}(x)$ and $\Pihom[F]_{ab}(x)$ respectively, i.e.
\begin{align*}
\FT{F}_{ef}(k)=\int_{x\in\Real^4} F_{ef}(x) e^{i k\cdot x} dx^{0123}
\end{align*}
and
\begin{align*}
\FTPiHom[F]_{ab}(k)=\int_{x\in\Real^4} \Pihom[F]_{ab}(x) e^{
  i k\cdot x} dx^{0123}
\end{align*}
where $k=k_a dx^a$, $k\cdot x=k_a x^a$. Similarly let $\FT
\chione_{ab}{}^{ef}(k)$ be the Fourier transformation of
$\tfrac12\epsilon^{cdef}\chione_{abcd}(x)$, i.e.
\begin{align}
\FT\chione_{ab}{}^{ef}(k)=
\tfrac12\epsilon^{cdef}\int_{x\in\Real^4}\chione_{abcd}(x) e^{i k\cdot
  x} dx^{0123}
\label{Disp_homo_chi_k}
\end{align}
If $\Zhom=0$ then it follows from (\ref{Disp_homo_conv}) that:
\begin{align}
\FTPiHom[F]_{ab}(k)=\tfrac12\FT\chione_{ab}{}^{cd}(k)\,\FT F_{cd}(k)
\label{Disp_homo_freq}
\end{align}
Since $\chi_{abcd}$ is a real function on $M$ its Fourier transform
satisfies
\begin{align*}
\cnj{\FT\chione_{ab}{}^{cd}(k)}=\FT\chione_{ab}{}^{cd}(-k)
\end{align*}
The 36 components of $\FT \chione_{ab}{}^{cd}(k)$ subject to this
symmetry can be expressed in terms of permittivity, permeability and
magneto-electric tensors relative to any observer frame.  A
specification of these components together with relations that
determine the electric current $J$ serve as an electromagnetic model
for a spacetime homogeneous medium in Minkowski spacetime.  If the medium lacks
this electromagnetic homogeneity recourse to the Fourier transform
(\ref{Disp_homo_conv}) is not possible and the constitutive properties
must be given in terms of a 2-point kernel and (\ref{Disp_coords_Pi}).

\section{Constitutive models for a collisionless ionized plasma}
\label{ch_Plasma}

As noted in the introduction the computation of the susceptibility for
homogeneous stationary dispersive media owes much to phenomenological
models and input from experiment. For certain conductors,
semi-conductors, insulators and low-dimensional structures much can
also be learnt from the application of quantum theory.  For
inhomogeneous and anisotropic media subject to non-stationary electromagnetic
fields linear responses are often the subject of a perturbation
approach.  This is particularly so in the case of ionized gases.

As an application of the above formalism the classical linear response
of a fully ionized inhomogeneous non-stationary collisionless plasma
to a perturbation is considered in the presence of an arbitrary
background gravitational field. The perturbed constitutive tensor will
be calculated in terms of solutions to the classical Maxwell-Vlasov
equations for the system.  This system is described in terms of the
electromagnetic 2-form $F\in\Gamma\Lambda^2 \Mp$ over a gravitational
spacetime $\Mp$, lying in the future of an initial hypersurface
$\SigmaM$, and a {\it collection of one-particle ``distribution''
  forms} (of degree 6), $\theta^\speciesa\in\Gamma\Lambda^6\Ebunp$
(one for each charged species of particle $\speciesabig$ with mass
$m^\speciesa$ and charge $q^\speciesa$) on the {\it upper unit
  hyperboloid bundle} $\pi:\Ebunp\to\Mp$ over $\Mp$.  The
$7$-dimensional manifold $\Ebunp$ is a sub-bundle of the
$8$-dimensional tangent bundle $T\Mp$ over $\Mp$ whose sections are
all future pointing time-like unit vector fields on $\Mp$. Thus
generic elements of $\Ebunp$ can be written $(z,w)$ with $z\in\Mp$ ,
$\pi(z,w)=z$ and $g(w,w)=-1$. The initial values of the one-particle
forms are given on the hypersurface $\SigmaE$ where
$\SigmaE=\pi^{-1}\Set{\SigmaM}\subset\Ebunp$.

The Maxwell-Vlasov system is usually written in terms of the Maxwell
system in vacuo and {\it all} sources are contained in the total current
$J\in\Gamma T\Mp$. This in turn is given by the sum over each species
current
\begin{align}
J=\sum_\speciesa J^\speciesa
\label{Plasma_Jf}
\end{align}
where $J^\speciesa\in\Gamma T \Mp$.  Thus in terms of $F$ and $J$ the
Maxwell subsystem is
\begin{align}
dF=0\qquadand \epsilon_0 d\star F=-\star \dual{J}
\label{Plasma_MV_Ff}
\end{align}

The dynamic equations for each $\theta^\speciesa$ can be written
succinctly in terms of forms on $\Ebunp$ and a collection of Liouville
vector fields $W^\speciesa \in \Gamma T\Ebunp$ describing the flow of
the charged particles associated with each species $[\alpha]$:
\begin{align}
W^\speciesa|_{(z,w)}= \Horiz_{(z,w)}(z,w)+
\frac{q^\speciesa}{m^\speciesa} \Vert_{(z,w)}(\dual{i_{(z,w)}F})
\label{Plasma_MV_W}
\end{align}
in terms of certain horizontal and vertical
lifts\cite{yano1973tangent}.  With these vector fields the
distribution forms $\theta^\speciesa$ are defined to satisfy the
collisionless conditions:
\begin{align}
d \theta^\speciesa=0 \end{align}
and
\begin{align} i_{W^\speciesa} \theta^\speciesa = 0
\label{Plasma_MV_theta}
\end{align}
To close this system one requires:
\begin{align}
\star\dual{J^\speciesa}= q^\speciesa\mapint_{\pi} \theta^\speciesa
\label{Plasma_Jf_sp}
\end{align}
The closure of $\theta^\speciesa$ leads, from
(\ref{Notation_mapint_comm_d}), to the continuity equation for each
species current:
\begin{align}
d\star\dual J^\speciesa = d\Big( \mapint_{\pi} \theta^\speciesa \Big)
= \mapint_{\pi} d \theta^\speciesa = 0
\label{Plasma_closed_current}
\end{align}
so the total current 3-form $\star\dual J$ is closed away from the
boundary $\SigmaM$.

A local coordinate system $(z^0,\ldots,z^3)$ for a region containing
$z$ on $\Mp$ induces a local coordinate system
$(z^0,\ldots,z^3,w^1,w^2,w^3)$ on $\Ebunp$.  Since $\Ebunp\subset T \Mp$
the tangent vector for a generic element $(z,w)\in\Ebunp$ may be
written
\begin{align*}
(z,w)=w^a\pfrac{}{z^a}\bigg|_z\in \Ebun^+_z\subset T_z \Mp
\end{align*}
where $\Ebun^+_z=\pi^{-1}(\Set{z})$ is the 3-dimensional fibre of
$\Ebunp$ over $z$ coordinated by $(w^1,w^2,w^3)$ and $w^0(z,w)$ is the
solution to $g_{ab}w^a w^b=-1$ with $w^0>0$.  All indices in the range
$0,1,2,3$ are raised and lowered using $g^{ab}$ and $g_{ab}$ so that
$w_0=w^a g_{a0}$.  Given a pair of vectors $(z,w),(z,v)\in
\Ebun^+_z\subset T_z\Mp$ the
horizontal lift of the vector $(z,v)$ to the point $(z,w)\in\Ebunp$
will be denoted $\Horiz_{(z,w)}(z,v)\in T_{(z,w)}\Ebunp$ and is given
by
\begin{align}
\Horiz_{(z,w)}(z,v) = \Big(v^a\pfrac{}{z^a}-\Gamma^\inu{}_{ef}(z)w^e
v^f \pfrac{}{w^\inu} \Big)\Big|_{(z,w)}
\label{Plasma_Horiz}
\end{align}
where $\Gamma^a{}_{ef}$ are the Christoffel symbols determined by the
metric components $g^{ab}$. Furthermore if $g(v,w)=0$ then the
vertical lift of the vector $(z,v)$ to the point $(z,w)\in\Ebunp$ is
given by
\begin{align}
\Vert_{(z,w)}(z,v) = \Big(v^\mu \pfrac{}{w^\mu}\Big)\Big|_{(z,w)} \in
T_{(z,w)}\Ebunp
\label{Plasma_Vert}
\end{align}
Thus from (\ref{Plasma_MV_W}), each Liouville vector field in these
coordinates can be expressed as
\begin{equation}
\begin{aligned}
W^\speciesa|_{(z,w)}&=
w^a\pfrac{}{z^a}+
\Big(-\Gamma^\inu{}_{ef}(z)w^e
w^f
+ \frac{q^\speciesa}{m^\speciesa}F_{ef}(z) g^{\inu e}
w^f\Big)\pfrac{}{w^\inu}
\end{aligned}
\label{Plasma_MV_W_coords}
\end{equation}
Denote by $\Omega\in\Gamma\Lambda^7\Ebunp$ the natural 7-form measure
on $\Ebunp$ given in these coordinates by
\begin{align}
\Omega=\frac{\detg}{w_0}dz^{0123}\wedge dw^{123}
\label{Plasma_Omega}
\end{align}
In ref. \onlinecite[eqn. (94)]{1971grc} it is shown that for all species
$\speciesabig$
\begin{align}
d i_{W^\speciesa} \Omega =0
\label{d_iW_Omega}
\end{align}
The {\it{distribution function}} $f^\speciesa\in\Gamma\Lambda^0\Ebunp$
relative to $\Omega$ for the species $\speciesabig$ is defined
implicitly via
\begin{align}
\theta^\speciesa=i_{W^\speciesa}(f^\speciesa\Omega)
\label{Plasma_f}
\end{align}
From (\ref{d_iW_Omega}, \ref{Plasma_f}) it follows that
(\ref{Plasma_MV_theta}) is equivalent to
\begin{align}
W^\speciesa(f^\speciesa)=0\,,
\label{Plasma_MV_f}
\end{align}
and from (\ref{Plasma_Jf_sp}) the components of the species current
$\speciesabig$ are given in terms of $f^\speciesa(z,w)$ by
\begin{align}
J^\speciesa{}^b(z)= q^\speciesa\int_{\Ebun_z^+}
\frac{w^b|(\det g)(z)|^{1/2}}{w_0(z,w)} f^\speciesa(z,w) dw^{123}
\label{Plasma_Jf_sp_coords}
\end{align}

\subsection{Perturbation analysis}

Let $\theta^\speciesa_1\in\Gamma\Lambda^6\Ebunp$ and $F_1\in\Gamma\Lambda^2 \Mp$
be  perturbations of $\theta^\speciesa_0$ and $F_0$, i.e.
\begin{align}
&\theta^\speciesa=\theta^\speciesa_0+\theta^\speciesa_1+\ldots
\quadand
F=F_0+ F_1+\ldots
\label{Plasma_MV_1}
\end{align}
where
\begin{equation}
\begin{aligned}
&d \theta^\speciesa_0=0 \,,\qquad
i_{W^\speciesa_0} \theta^\speciesa_0 = 0\,,\qquad
\\&
dF_0=0 \,,\qquad
\epsilon_0 d\star F_0=
-\sum_\speciesa q^\speciesa \mapint_{\pi} \theta^\speciesa_0
\end{aligned}
\label{Plasma_MV0}
\end{equation}
 and
\begin{align}
W^\speciesa_0|_{(z,w)}=
\Horiz_{(z,w)}(z,w) +
\frac{q^\speciesa}{m^\speciesa}\Vert_{(z,w)}\big(\dual{i_{(z,w)}F_0}\big)
\label{Plasma_MV_W0}
\end{align}
i.e. given by substituting $F=F_0$ into (\ref{Plasma_MV_W_coords}).
Substituting $F$ into (\ref{Plasma_MV_W}) yields
$W^\speciesa=W^\speciesa_0+ W^\speciesa_1+\ldots$ where $W^\speciesa_1=\What^\speciesa_1(F_1)$
and the map
$\What_1:\Gamma\Lambda^2 \Mp\to\Gamma T\Ebunp$
is given by
\begin{equation}
\What^\speciesa_1(F_1)|_{(z,w)}=\frac{q^\speciesa}{m^\speciesa}
\Vert_{(z,w)}\big(\dual{i_{(z,w)}F_1}\big)
\label{Plasma_MV_W1}
\end{equation}
The first order linear system for the perturbation
$(\theta_1,F_1)$ is then
\begin{align}
d \theta^\speciesa_1&=0 \,,
\label{Plasma_MV_d_theta1}
\\
i_{W^\speciesa_0} \theta^\speciesa_1
&= - i_{\What^\speciesa_1(F_1)} \theta^\speciesa_0\,,
\label{Plasma_MV_iW_theta1}
\\
dF_1&=0\,,
\label{Plasma_MV_dF1}
\\
\epsilon_0 d\star F_1&= -\sum_\speciesa q^\speciesa
\mapint_{\pi} \theta^\speciesa_1
\label{Plasma_MV_dstarF1}
\end{align}
Using (\ref{Notation_mapint_comm_d}) and (\ref{Plasma_MV_d_theta1}) it
follows that each species current in the sum on the right hand side of
(\ref{Plasma_MV_dstarF1}) is closed away from the initial hypersurface
$\SigmaM$.  In terms of the excitation field
$G_1\in\Gamma\Lambda^2\Mp$ equation (\ref{Plasma_MV_dstarF1}) will be
written
\begin{align}
d\star G_1 = 0
\label{Plasma_MV_dstarG1}
\end{align}
where
\begin{align}
G_1 = \epsilon_0 F_1 + \Pi_1[F_1,\zetabd_1]
\label{Plasma_MV_G1}
\end{align}
for some linear functional $\Pi_1$ of $F_1$ and $\zetabd$ such that
\begin{align}
d\star \Pi_1[F_1,\zetabd_1]= -\sum_\speciesa \mapint_{\pi}
\theta^\speciesa_1
\label{Plasma_Pi1}
\end{align}
and
$\zetabd_1=\Set{\zeta^{\speciesas{1}}_1,\zeta^{\speciesas{2}}_1,\ldots}$
where
$\zeta^\speciesa_1=\xi^\speciesa_1|_{\SigmaEY}$
for some $\xi^\speciesa_1\in\Gamma\Lambda^5\EbunYp$ which solves
$\theta^\speciesa_1=d\xi^\speciesa_1$. Thus
$\zeta^\speciesa_1$ is related to the initial velocity profile of the
species $\speciesabig$.

In the next section \ref{ch_SusKernel} the general susceptibility kernel
$\chi\in\Gamma\Lambda^0(\MXp\times\MYp)$ and linear
functional $Z_1$, determined by $\theta^\speciesa_0$ and
$F_0$, are found such that
\begin{align}
\Pi_1[F_1,\zetabd_1]|_x=\mapint_{\PiX}
\chi\wedge \PiY^\star(F_1)
+Z_1[\zetabd_1]
\label{Plasma_susep}
\end{align}
satisfies (\ref{Plasma_Pi1}).

\subsection{A general formula for the functional $\Pi_1$ in an
  unbounded plasma}
\label{ch_SusKernel}

\begin{figure}
\setlength{\unitlength}{0.09\textwidth}
\begin{picture}(5,2)
\put(0.1,0.1){\includegraphics[width=5\unitlength,viewport=0 0 507 207]{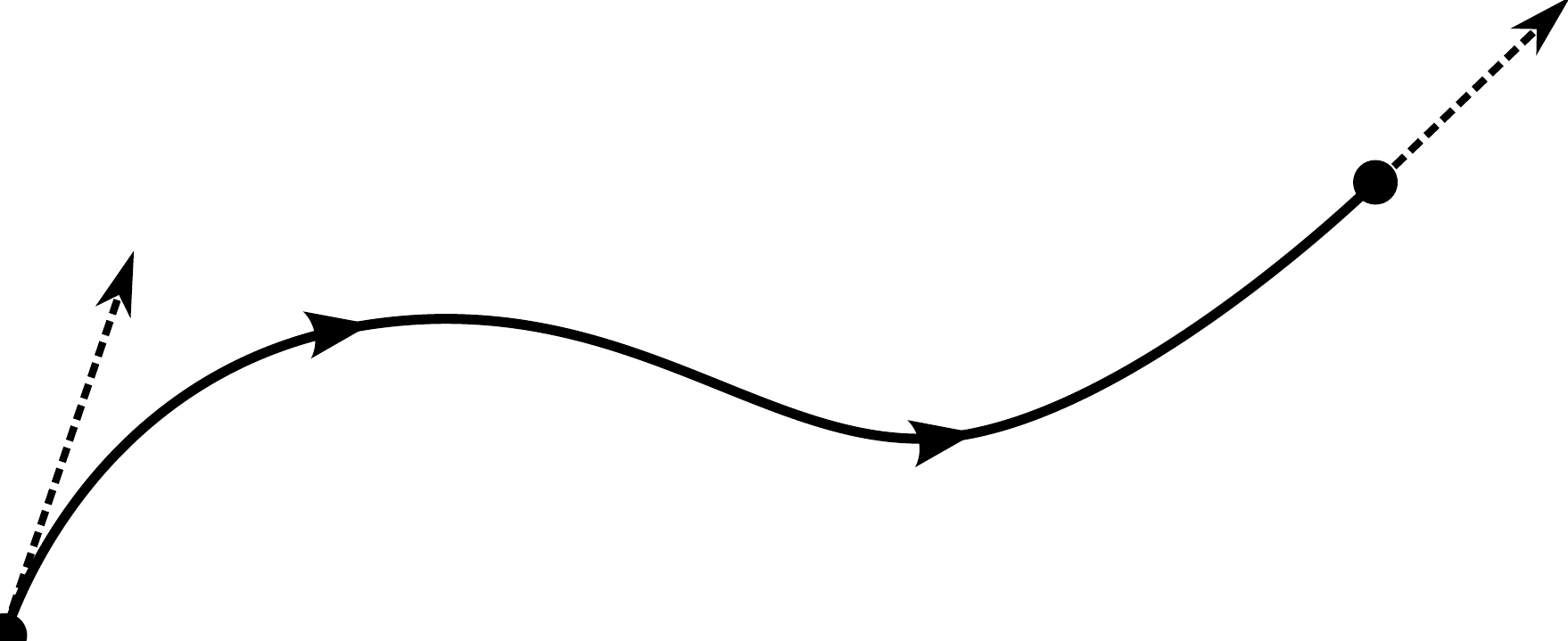}}
\put(0,-0.1){$y$}
\put(0.2,1){$u$}
\put(2.5,1){$\tau$}
\put(4.5,1.3){$x$}
\put(5.0,2.2){$v$}
\end{picture}
\caption{{A segment of the solution curve $C_{(x,v)}$ to the
    unperturbed Lorentz force equation (\ref{Plasma_Lorentz_Force})
    with final position $x$, final velocity $(x,v)$, initial position
    $y=C_{(x,v)}(\tau)$ and initial velocity $(y,u)=\Cd_{(x,v)}(\tau)$.}}
\label{fig_world_line}
\end{figure}
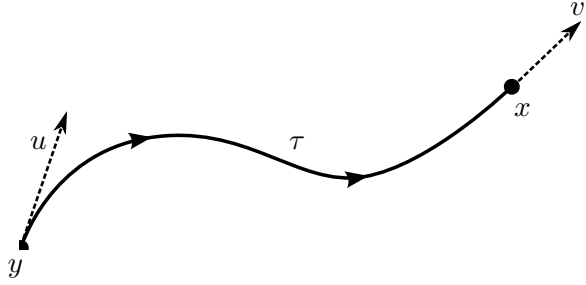

In this section a general expression for a susceptibility kernel will
be constructed in terms of the integral curves of the vector field
$W^{[\alpha]}_0\in\Gamma T\Ebunp$. Such curves describe segments of
particle world lines under the influence of the Lorentz force due to
the external electromagnetic field $F_0$.  Although, for a general
$F_0$, it is not possible to derive an analytic form for such integral
curves, special cases are amenable to an analytic analysis.

It proves convenient to let the final and initial states of each
species of particle reside in fibres over $\MXp$ and $\MYp$
respectively, bounded by the equivalent hypersurfaces
$\SigmaMX\subset\MXp$ and $\SigmaMY\subset\MYp$.  Thus the
corresponding upper unit hyperboloid bundles $\piX:\EbunXp\to\MXp$ and
$\piY:\EbunYp\to\MYp$ with boundary hypersurfaces
$\SigmaEX\subset\EbunXp$ and $\SigmaEY\subset\EbunYp$ are used to
accommodate the final and initial 4-velocities of the particles.
The generic
elements of these bundles are written $(x,v)\in\EbunXp$ and
$(y,u)\in\EbunYp$ where $x\in\MXp$, $y\in\MYp$ and $g(v,v)=g(u,u)=-1$.
The induced coordinate systems for $\EbunXp$ and $\EbunYp$ are
$(x^0,\ldots,x^3,v^1,v^2,v^3)$ and $(y^0,\ldots,y^3,u^1,u^2,u^3)$. Let
$v^0(x,v)$, $v_0(x,v)$, $u^0(y,u)$ and $u_0(y,u)$ be defined in the
same way as $w^0(z,w)$ and $w_0(z,w)$.

The contribution to the tensor $\Pi_1[F_1,\zetabd_1]$ due to all
dynamic sources,
arises from all particle histories in the past light cone of
$x\in\MXp$. The history of the species particle $\speciesabig$ which
passes through event $x$ with 4-velocity $(x,v)\in\EbunXp$ will
therefore be parametrized by negative proper time $\tau$:
$C^\speciesa_{(x,v)}:[\tau^\speciesa_0(x,v),0]\to\Mp$,
$\tau\mapsto C^\speciesa_{(x,v)}(\tau)$.
Such a history is the unique solution to the Lorentz force equation
\begin{equation}
\nabla_{\Cd^\speciesa_{(x,v)}}\Cd^\speciesa_{(x,v)}=
\frac{q^\speciesa}{m^\speciesa}\big(\dual{i_{\Cd^\speciesa_{(x,v)}} F_0}\big)
\label{Plasma_Lorentz_Force}
\end{equation}
with
\begin{equation}
g(\Cd^\speciesa_{(x,v)},\Cd^\speciesa_{(x,v)})=-1
\label{Plasma_Lorentz_Force_norm}
\end{equation}
and final condition
\begin{equation}
C^\speciesa_{(x,v)}(0)=x\,,\
\Cd^\speciesa_{(x,v)}(0)=(x,v)
\label{Plasma_Lorentz_Force_fin}
\end{equation}
where
$\Cd^\speciesa_{(x,v)}(\tau)=C^\speciesa_{(x,v)\star}(\partial_\tau|_\tau)
=\Cd^{\speciesa a}_{(x,v)}(\tau)
\pfrac{}{x^a}$ and the value $\tau^\speciesa_0(x,v)\le0$ solves
\begin{align}
C^\speciesa_{(x,v)}\big(\tau^\speciesa_0(x,v)\big)\in\SigmaMY
\label{Plasma_Lorentz_Force_Sigma}
\end{align}
This defines the prolongation of $C$,
$\Cd^\speciesa_{(x,v)}:[\tau^\speciesa_0(x,v),0]\to\Ebunp$.
For each species $\speciesabig$, $(x,v)\in\EbunXp$
and $\tau\in[\tau^\speciesa_0(x,v),0]$ let
$(y,u)\in\EbunYp$ denote the initial state, i.e.
$y=C^\speciesa_{(x,v)}(\tau)$ and $(y,u)=\Cd^\speciesa_{(x,v)}(\tau)$,
see figure \ref{fig_world_line}.

The {\it family} of {\it all} such histories is described in terms of
the maps
\begin{align}
\phi^\speciesa:\Nbun^\speciesa_X\to\EbunYp
\,,\qquad
\phi^\speciesa(\tau,x,v)=\Cd^\speciesa_{(x,v)}(\tau)
\label{Plasma_def_phi}
\end{align}
where
\begin{align*}
\Nbun^\speciesa_X=\Set{(\tau,x,v)\in\Real^-\times\EbunXp\,\big|\,
\tau^\speciesa_0(x,v)\le\tau\le0}
\end{align*}
The manifold $\Nbun^\speciesa_X$ with boundary  is naturally a fibre bundle
over $\EbunXp$ with projection
$\ppX^\speciesa:\Nbun^\speciesa_X\to\EbunXp$,
$(\tau,x,v)\mapsto\ppX^\speciesa(\tau,x,v)=(x,v)$ and for any form
$\alpha\in\Gamma\Lambda^p\Nbun_X$ it follows from
(\ref{Disp_def_fibre_int}) that
\begin{align*}
\mapint_{\ppX^\speciesa}\alpha =
dx^I \wedge dy^J\int^0_{\tau^\speciesa_0(x,v)} \alpha^{(1)}(\tau,x,v) d\tau
\end{align*}
where $\alpha=\alpha^{(1)}(\tau,x,v)
dx^I \wedge dy^J\wedge d\tau+\alpha^{(2)}(\tau,x,v)
dx^I \wedge dy^J$.

Let $\Gamma\Lambda^5_{\SigmaEY}\EbunYp$ be the set of sections over
$\SigmaEY$ with values in $\Lambda^5\EbunYp$,
i.e. if $\alpha\in\Gamma\Lambda^5_{\SigmaEY}\EbunYp$ then
for each
$(y,u)\in\SigmaEY$,
$\alpha|_{(y,u)}\in\Lambda^5_{(y,u)}\EbunYp$.
Let the map
$\varphi^\speciesa:\Gamma\Lambda^5_{\SigmaEY}\EbunYp
\to\Gamma\Lambda^5\EbunXp$ be given by
\begin{align}
\varphi^\speciesa(\alpha)|_{(x,v)}=
\phi^{\speciesa\star}_{\tau^\speciesa_0(x,v)}(\alpha|_{\tau^\speciesa_0(x,v)})
\in\Lambda^5_{(x,v)}\EbunXp
\label{SusKernel_def_varphi}
\end{align}
where
$\phi^{\speciesa}_{\tau}:\EbunXp\to\Ebunp$,
$\phi^{\speciesa}_{\tau}(x,v)=\phi(\tau,x,v)$.
For each species $\speciesabig$ let the initial data be given by
$\zeta^\speciesa_{1}\in\Gamma\Lambda^5_{\SigmaEY}\EbunYp$ with
$i_{W_0^\speciesa}\zeta^\speciesa_{1}=0$.

In terms of these maps, it will now be shown that the general polarization
functional $\Pi_1$ on $\MXp$ is given by
\begin{equation}
\begin{aligned}
\Pi_1[F_1,\zetabd_1]
&=
\sum_\speciesa {q^{\speciesa}}
\star \mapint_{\piX}\mapint_{\ppX^\speciesa} d\tau \wedge
\phi^{\speciesa\star} (i_{\What^\speciesa_1(F_1)} \theta^\speciesa_0)
+
\star d\big( \Xi_1[F_1]\big)
\\&\quad
+
\sum_\speciesa {q^{\speciesa}}
\star \mapint_{\piX} \varphi^\speciesa(\zeta^\speciesa_{1})
+
\star d\big( \Zonez[\zetabd_1]\big)
\end{aligned}
\label{Plasma_Pi_res}
\end{equation}
where $\Xi_1$ and $\Zonez$ are arbitrary linear functionals of $F_1$ and
$\zetabd_1$  respectively. The excitation
$\Pi_1[F_1,\zetabd_1]$, in (\ref{Plasma_Pi_res}), is the general solution to
(\ref{Plasma_Pi1}) where the source $\theta_1$ satisfies
{(\ref{Plasma_MV_d_theta1},\ref{Plasma_MV_iW_theta1})}.  The first two
terms on the right hand side of (\ref{Plasma_Pi_res}) are linear
functionals of $F_1$ whereas the last term is a linear functional of
the initial data $\zetabd_1$.  Clearly $\star
d\big(\Xi_1[F_1]\big)$ and $\star
d\big(\Zonez[\zetabd_1]\big)$ are in the kernel of $d
\star$, the homogeneous differential operator associated with
(\ref{Plasma_Pi1}).

The proof that (\ref{Plasma_Pi_res}) solves
(\ref{Plasma_Pi1}) requires the following lemma which is proved in the
appendix.

\begin{lemma}
\label{lm_int_chac_curves}
Let $N$ be a manifold with a boundary $\SigmaN\subset N$ and
let $V\in\Gamma TN$ be a non-vanishing vector field on $N$ such that
every integral curve of $V$ intersects $\SigmaN$ precisely once.
For each $\sigma\in N$ let the integral curve of $V$ terminating
at $\sigma$ be given by $\gamma_\sigma:[\tau_0(\sigma),0]\to N$
where $\gamma_\sigma(0)=\sigma$ and $\gamma_\sigma(\tau_0(\sigma))\in\SigmaN$.
The set $\Nbun=\Set{(\sigma,\tau)\subset\Real^-\times
  N\,\big|\,\tau_{\min}(\sigma)\le\tau\le 0}$ is a fibred manifold over $N$ with
projection $\varpi_N:\Nbun\to N$,
$(\tau,\sigma)\mapsto\varpi_N(\tau,\sigma)=\sigma$. The family of
integral curves of $V$ can be described by the map
$\phi_N:\Nbun\to N$,
$\phi_N(\tau,\sigma)=\gamma_\sigma(\tau)$.
Let $\zeta\in\Gamma\Lambda^p_{\SigmaN}N$ such that $i_V\zeta=0$,
i.e. $\zeta$ is a $p$-form on $\SigmaN$ with values in $\Lambda^p
N$. Let $\varphi_{N}:\Gamma\Lambda^p_{\SigmaN}N\to\Gamma\Lambda^pN$
be given by
$\varphi_{N}(\zeta)|_\sigma=
\phi^\star_{N\,\tau_0(\sigma)}(\zeta|_{\tau_0(\sigma)})\in\Lambda^p_{\SigmaN} N$.

If $\beta\in\Gamma\Lambda^p N$ is a $p$-form on $N$ with compact
support such that $i_V\beta=0$ and $\xi\in\Gamma\Lambda^p N$ has the
form
\begin{align}
\xi
=
\mapint_{\varpi_N} \phi_N^{\star}(\beta) \wedge d\tau
+
\varphi_{N}(\zeta)
\label{thepde_alpha}
\end{align}
then
\begin{align}
i_V d \xi = \beta
\label{thepde_iVd_alpha}
\end{align}
and $\xi|_{\SigmaN}=\zeta$.
\end{lemma}

This lemma is applied with $N=\EbunXp$,
$\varpi_N=\ppX^\speciesa$, $V=W^\speciesa_0$, $\tau_0=\tau^\speciesa_{0}$,
$\phi_N=\phi^\speciesa$, $\varphi_{N}=\varphi^\speciesa$,
$\zeta=\zeta^\speciesa_{1}$ and
\begin{align}
\beta=-i_{\What^\speciesa_1(F_1)}\theta^\speciesa_0
\label{SusKernel_beta_replace}
\end{align}
Thus $\xi$
in (\ref{thepde_alpha}) becomes the 5-form
$\xi^\speciesa_1\in\Gamma\Lambda^5\EbunXp$,
\begin{align}
\xi^\speciesa_{1}
=
-\mapint_{\ppX^\speciesa}
\phi^{\speciesa\star}\big(i_{\What^\speciesa_1(F_1)}\theta^\speciesa_0\big)
\wedge d\tau
+\varphi^\speciesa(\zeta^\speciesa_{1})
=
\mapint_{\ppX^\speciesa} d\tau\wedge
\phi^{\speciesa\star}\big(i_{\What^\speciesa_1(F_1)}\theta^\speciesa_0\big)
+\varphi^\speciesa(\zeta^\speciesa_{1})
\label{SusKernel_def_zeta1}
\end{align}
since
$\deg\big(\phi^{\speciesa\star}(i_{\What^\speciesa_1(F_1)}
\theta^\speciesa_0)\big)=5$.
In order to satisfy (\ref{Plasma_MV_d_theta1}) let
\begin{align}
\theta_1^\speciesa=d\xi^\speciesa_{1}
\label{SusKernel_res_theta1}
\end{align}
Furthermore from (\ref{thepde_iVd_alpha}) and (\ref{SusKernel_beta_replace})
\begin{align*}
i_{W^\speciesa_0}\theta_1^\speciesa
=
i_{W^\speciesa_0}d\xi^\speciesa_{1}
=
-i_{\What^\speciesa_1(F_1)}\theta^\speciesa_0
\end{align*}
so (\ref{Plasma_MV_iW_theta1}) is satisfied. In terms of
$\xi^\speciesa_{1}$ (\ref{Plasma_Pi_res}) can be written
\begin{align*}
&\Pi_1[F_1,\zetabd_1]|_x
=
\sum_\speciesa {q^{\speciesa}}
\star \mapint_{\piX}\xi^\speciesa_{1}
+
\star d\big( \Xi_1[F_1]\big)|_x
+
\star d\big( \Zonez[\zetabd_1]\big)
\end{align*}
Then from (\ref{Notation_mapint_comm_d})
\begin{align*}
d\star\Pi_1[F_1,\zetabd_1]
&=
 d\star\star\Big(\sum_\speciesa q^\speciesa
\mapint_{\piX}  \xi^\speciesa_{1}\Big)
\\&=
-\sum_\speciesa q^\speciesa
d \mapint_{\piX} \xi^\speciesa_{1}
=
-\sum_\speciesa q^\speciesa \mapint_{\piX} d\xi^\speciesa_{1}
\\&=
- \sum_\speciesa q^\speciesa \mapint_{\piX} \theta^\speciesa_1
\end{align*}
Thus the Maxwell equation (\ref{Plasma_Pi1}) is also satisfied. That
(\ref{Plasma_Pi_res}) is the general solution to (\ref{Plasma_Pi1})
follows from the fact that the difference between any two solutions of
(\ref{Plasma_Pi1}) satisfies the homogeneous differential equation
associated with (\ref{Plasma_Pi1}).

Thus we have succeeded in eliminating $\theta_1^\speciesa$ from the
perturbation system
(\ref{Plasma_MV_d_theta1}-\ref{Plasma_MV_dstarF1}), thereby reducing
the system to $dF_1=0$ and
\begin{equation}
\begin{aligned}
&\epsilon_0 d\star F_1
+
\tfrac12
\sum_\speciesa q^\speciesa
d
\mapint_{\pi_X}\mapint_{\ppX^\speciesa} d\tau \wedge
\phi^{\speciesa\star} (i_{\What^\speciesa_1(F_1)}
\theta^\speciesa_0)
+
\sum_\speciesa {q^{\speciesa}}
d \mapint_{\piX} \varphi^\speciesa(\zeta^\speciesa_{1})
=0
\end{aligned}
\label{Plasma_disp_eqn}
\end{equation}
in terms of $(\theta_0,F_0)$, for the perturbation $F_1$. The
perturbation $\theta_1$ is then given by
(\ref{SusKernel_res_theta1},\ref{SusKernel_def_zeta1}).


\subsection{The susceptibility kernel for an unbounded collisionless plasma}

Equating (\ref{Plasma_Pi_res}) and (\ref{Plasma_susep})
with the initial data
\begin{align}
Z_1[\zetabd_1]=\sum_\speciesa {q^{\speciesa}}
\star \mapint_{\piX} \varphi^\speciesa(\zeta^\speciesa_{1})
+
\star d\big( \Zonez[\zetabd_1]\big)
\label{SusKernel_Z1_soln}
\end{align}
yields
\begin{equation}
\begin{aligned}
&\mapint_{\PiX}
\chi\wedge \PiY^\star(F_1)
=
\sum_\speciesa {q^{\speciesa}}
\star \mapint_{\piX}\mapint_{\ppX^\speciesa} d\tau \wedge
\phi^{\speciesa\star} (i_{\What^\speciesa_1(F_1)} \theta^\speciesa_0)
+
\star d\big( \Xi_1[F_1]\big)
\end{aligned}
\label{SusKernel_Pi1_eq_xi}
\end{equation}
Away from the initial hypersurface boundary
$\partial(\MXp\times\MYp)=\SigmaMX\times\MYp\union\MXp\times\SigmaMY$, using
(\ref{Notation_mapint_comm_d}) and (\ref{Hodge_int_commute}) one has
\begin{align*}
&\mapint_{\PiX} \star_X d_X\check\xi \wedge \PiY^\star(F_1)
=
\mapint_{\PiX} \star_X d\check\xi \wedge \PiY^\star(F_1)
=
\star d\mapint_{\PiX} \check\xi \wedge \PiY^\star(F_1)
=
\star d\big(\check\Xi_1[F_1]\big)
\end{align*}
where $\check\Xi_1 [F_1]$ is a linear functional of $F_1$.  The
gauge freedom $\chi\to\star_X d_X\check\xi$ given in
(\ref{Disp_gauge_dstardYchi}) is equivalent to the addition of the
term $\star d\big( \Xi_1[F_1]\big)$ in (\ref{Plasma_Pi_res}).

If $F_1$ is restricted to have support in a certain
domain one may find $\chi$ such that
\begin{equation}
\begin{aligned}
&\mapint_{\PiX}
\chi\wedge \PiY^\star(F_1)
=
\sum_\speciesa {q^{\speciesa}}
\star \mapint_{\piX}\mapint_{\ppX^\speciesa} d\tau \wedge
\phi^{\speciesa\star} (i_{\What^\speciesa_1(F_1)} \theta^\speciesa_0)
\end{aligned}
\label{SusKernel_Pi1_eq}
\end{equation}
To find such a susceptibility kernel requires the following maps.

For $(y,u)\in\EbunYp$, let $C^\speciesa_{(y,u)}:\Real^+\to\Mp$ and
$\Cd^\speciesa_{(y,u)}:[0,\tau^\speciesa_1(y,u)\big)\to\Ebunp$ be the unique
solutions to the unperturbed Lorentz force equation
(\ref{Plasma_Lorentz_Force},\ref{Plasma_Lorentz_Force_norm}) with
initial conditions
\begin{equation}
C^\speciesa_{(y,u)}(0)=y
\qquadand
\Cd^\speciesa_{(y,u)}(0)=(y,u)
\label{Plasma_Lorentz_Force_init}
\end{equation}
where $\tau^\speciesa_1(y,u)\in\Real^+\union\Set{\infty}$ is the
supremum of the values of $\tau$ such that
$C^\speciesa_{(y,u)}(\tau)\in M$. Let
$\Phi^\speciesa:\Nbun^\speciesa_Y\to \MXp\times \MYp$,
\begin{align}
\Phi^\speciesa(\tau,y,u)=\big(C^\speciesa_{(y,u)}(\tau),y\big)
\label{Plasma_def_varphi}
\end{align}
where
\begin{align*}
\Nbun^\speciesa_Y=\Set{(\tau,y,u)\in\Real^+\times\EbunXp\,\big|\,
0\le\tau<\tau^\speciesa_1(y,u)}
\end{align*}
This map gives the final and initial positions of a solution to
the unperturbed Lorentz force equation in terms of the
initial position, velocity and proper time parameter
$\tau\in[0,\tau^\speciesa_1(y,u)\big)$.

Observe that $\Phi^\speciesa$ is never surjective, since if
$\Phi^\speciesa(\tau,y,u)=(x,y)$ then $x\in J^+(y)$. Also
$\Phi^\speciesa$ is never injective since $\Phi^\speciesa(0,y,u)=(y,y)$
for all $(y,u)\in\EbunYp$.  Thus $\Phi^\speciesa$ does not possess an inverse
and one must work locally on $\MXp\times\MYp$ in order to establish the
diffeomorphism $\Psi^\speciesa:\DomPsi \to \DomPsi'$,
\begin{align}
\Psi^\speciesa=\left(\Phi^\speciesa|_{\DomPsi'}\right)^{-1}
\label{Plasma_def_psi}
\end{align}
i.e.
\begin{align*}
\Psi^\speciesa\big(C_{(y,u)}(\tau),y\big)=(\tau,y,u)
\end{align*}
with $\DomPsi\subset\MXp\times\MYp$ and
$\DomPsi'\subset\Nbun^\speciesa_Y$ given by
\begin{align}
\DomPsi
&=\Set{(x,y)\,\bigg|\,
\text{\parbox{0.84\textwidth}{There exists a unique $u\in\Ebun_y$ and
    $\tau\in\Real^+$ such that $C^\speciesa_{(y,u)}(\tau)=x$ for all
    $\speciesabig$}}}
\label{Plamsa_def_D}
\end{align}
and
\begin{align*}
\DomPsi'&=
\Set{(\tau,y,u)\,\Big|\,\Phi^\speciesa(\tau,y,u)\in\DomPsi
  \text{ for all $\speciesabig$}}
\end{align*}
This map $\Psi^\speciesa$ encodes the solution to the two-point
problem, namely given an initial event $y\in\MY$ and final event
$x\in\MX$ find the unique worldline to the unperturbed Lorentz force
equation which passes though these two points. This worldline is
specified by its initial velocity $(y,u)\in\EbunXp$ and its proper time
$\tau$. The statement that $\Phi^\speciesa$ does not have an inverse
is equivalent to the statement that in general there may not be a
unique solution to the two point problem on an arbitrary
domain.  The domain $\DomPsi$ is the set of all pairs $(x,y)$
such that there is a unique worldline.

Set
\begin{align}
\chi=\sum_\speciesa \chi^\speciesa
\label{Plasma_chi_sum_species}
\end{align}
where
\begin{equation}
\begin{aligned}
&\chi^\speciesa|_{(x,y)}=
\tfrac1{2} \frac{q^{\speciesa\,2}}{m^\speciesa}\hashx
dy^{cd}
\wedge\IC{y}{abcd}
\Psi^{\speciesa\star} \Big( d\tau \wedge
\varpi_Y^{\speciesa\star}
\big(g^{\nu a} u^b \IC{u}{\nu} \theta^\speciesa_0\big)\Big)\Big|_{(x,y)}
\end{aligned}
\label{Plasma_res_chi}
\end{equation}
for points $(x,y)\in\DomPsi$.
In the appendix (lemma \ref{lm_chi_nocoords}) it is shown that
given $x\in\MXp$
(\ref{SusKernel_Pi1_eq}) and $F_1$ with support in
\begin{align}
\DomPsi_x=\DomPsi\inter\PiX^{-1}\Set{x}= \Set{y\in\MY|(x,y)\in\DomPsi}
\label{SusKernel_def_DomPsi_X}
\end{align}
then (\ref{SusKernel_Pi1_eq}) holds at $x$.
Furthermore although
$(d_Y \chi)|_{(x,y)}$ is unique, $\chi$ has the gauge
freedom given by (\ref{Disp_gauge_dYchi}).

One may write (\ref{Plasma_res_chi})
implicitly as
\begin{align}
\chi^\speciesa\wedge\PiY^\star \gamma
=
-q^\speciesa
\hashx S \Psi^{\speciesa\star} \Big( d\tau \wedge
\ppY^{\speciesa\star} (i_{\What^\speciesa_1(\gamma)} \theta^\speciesa_0)\Big)
\label{Plasma_res_chi_imp}
\end{align}
for all
$\gamma\in\Gamma\Lambda^2\MYp$
where
$S:\Lambda^6_{(x,y)}(\MXp\times\MYp)\to \Lambda^6_{(x,y)}(\MXp\times\MYp)$,
\begin{align}
S(\alpha)=\IC{y}{0123}\alpha\wedge dy^{0123}
\label{Plasma_def_S}
\end{align}
The tensor projector $S$ has the simplest representation in the
coordinate basis employed here since $\IC{y}{a}dy^b=\delta^b_a$.

From (\ref{Plasma_def_psi}) for a chosen species $\speciesabig$ one
must consider $\tau$ and $u$ to be functions of $(x,y)$ as well as the
species label $\speciesabig$. Thus let $\Psi^\speciesa$ be given by
the functions $\tau=\tau(x,y)$ and $u^\mu=u^\imu(x,y)$, where we have dropped the
species label, i.e. $\tau(x,y)$ and $u^\imu(x,y)$ solve the implicit
equation
\begin{align}
C^\speciesa_{(y,u(x,y))}\big(\tau(x,y)\big) = x
\label{SusKernel_tau_u_implicit_eqn}
\end{align}
where $u^0(x,y)$ is the solution to $u^a(x,y) u^b(x,y) g_{ab}(y)=-1$
and $u_0(x,y)=g_{a0}(y) u^a(x,y)$.  Let $f_0^\speciesa=f_0^\speciesa(y,u)$
represent the unperturbed probability function on $\EbunYp$. The
contribution to the susceptibility kernel from species $\speciesabig$
is given in local coordinates by (lemma \ref{lm_chi_coords} in appendix.)
\begin{equation}
\begin{aligned}
{\chi^\speciesa|_{(x,y)}}
&=
-f_0^\speciesa \frac{q^{\speciesa2}}{m^\speciesa}
\frac{\detg^{3/2}}{4 u_0} g^{\imu c} u^b \epsilon^{dejk}
\epsilon_{cbih}
\epsilon_{\imu\inu\isigma}
\times
\\&\qquad
\bigg(
\frac{u^a}{2}\pfrac{\tau}{y^a}\pfrac{u^\inu}{x^d}\pfrac{u^\isigma}{x^e}
-\frac{u^a}{2}\pfrac{\tau}{x^d}\pfrac{u^\inu}{y^a}\pfrac{u^\isigma}{x^e}
+\frac{u^a}{2}\pfrac{\tau}{x^d}\pfrac{u^\inu}{x^e}\pfrac{u^\isigma}{y^a}&&
\\&\qquad \qquad +
\big(-\Gamma^\inu{}_{pf}u^{p} u^f + \frac{q^\speciesa}{m^\speciesa}
F_{0pf} g^{\inu p} u^f\big)
\pfrac{\tau}{x^d}\pfrac{u^\isigma}{x^e}
\bigg) dx_{jk}\wedge dy^{ih}
\end{aligned}
\label{Plasma_coords}
\end{equation}
where $g$, $F_0$ and $\Gamma^\inu{}_{ef}$ are all evaluated at $y\in
\MYp$ and each $\tau$ and $u$ belongs to the species $\speciesabig$.
This is a key result of our article.

\subsection{A spacetime inhomogeneous microscopically neutral plasma.}
\label{sch_coords_neut}

In a Vlasov model, a plasma or gas is deemed \defn{microscopically
  neutral} if in its unperturbed state $F_0=0$.  Let $M$ be Minkowski
spacetime with global Lorentzian coordinates so that
$\Gamma^\inu_{ab}=0$.  Assume that $f_0^\speciesa$ solves the zeroth
order Maxwell-Vlasov system (\ref{Plasma_MV0}) with
$\theta^\speciesa_0=i_{W^\speciesa_0}(f^\speciesa_0\Omega)$ and
$F_0=0$. In this scenario one can calculate $\chi$ explicitly.

Since Minkowski spacetime is flat and $F_0=0$ the integral curves
$C_{(x,v)}$ in global Lorentzian coordinates are the straight lines:
\begin{align}
\tau=\sqrt{-g(x-y,x-y)}
\qquadand
u=\frac{(x-y)}{\tau}
\label{Plasma_Eg_tau_u}
\end{align}
Differentiating with respect to $x^a$ and $y^a$ gives.
\begin{equation}
\begin{aligned}
&\pfrac{\tau}{x^a}=-u_a\,,\qquad
\pfrac{\tau}{y^a}=u_a\,,\qquad
\pfrac{u^a}{x^b}=\frac{(\delta^a_b+u_a u_b)}{\tau}
\\&
\qquadand
\pfrac{u^a}{y^b}=-\frac{(\delta^a_b+u_a u_b)}{\tau}
\end{aligned}
\label{Plasma_Eg_diff}
\end{equation}
If follows from (\ref{Plasma_coords}) that
\begin{equation}
\begin{aligned}
\chi^\speciesa|_{(x,y)}&=
\frac{q^\speciesa f_0^\speciesa(y,u)}{4 u_0\tau^2} g^{\imu c} u^b
\epsilon_{cbih}
\big(2dx_{0\mu}+
\epsilon^{d\sigma jk}
\epsilon_{\imu\inu\isigma}
u^\inu u_d dx_{jk}
\big)
\wedge dy^{ih}
\end{aligned}
\label{Plasma_Eg_res}
\end{equation}
where $\tau(x,y)$ and $u(x,y)$ are given by (\ref{Plasma_Eg_tau_u}).

It is often useful to explore the response of an inhomogeneous
plasma due to a monochromatic electromagnetic plane wave with constant
amplitude $E$:
\begin{align}
F_1=E e^{-i\omega x^0+i\kk x^1} dx^{01}\,.
\label{Plasma_plane_wave}
\end{align}
Setting the initial hypersurface as $\SigmaEY=\Set{y^0=y^0_0}$, the
general initial 5-form $\zeta^\speciesa_1\in\Gamma\Lambda^5_{\SigmaEY}
\EbunYp$ satisfying $i_{W_0}\zeta^\speciesa_1=0$ is given in terms of
its components by
\begin{equation}
\begin{aligned}
\zeta^\speciesa_1|_{(0,y^\mu,u^\nu)}
&=
\big(u^0 dy^{1}-{u^1}dy^{0}\big)\wedge
\big(
\zetao{1}dy^2\wedge du^{123}+ \zetao{2}dy^3\wedge du^{123}
\big)
+ \zetao{3} dy^{23}\wedge du^{123}
\\&
+\big(u^0 dy^{123}-{u^1}dy^{023}\big)
\big(
\zetao{4}du^{12} + \zetao{5}du^{13}+ \zetao{6}du^{23}
\big)
\end{aligned}
\label{Plasma_Eg_generic_zeta}
\end{equation}
where $\zetao{A}=\zetaoh{A}(y^\mu,u^\nu)$ for $A=1,\ldots 6$.
For the integral curves (\ref{Plasma_Eg_tau_u}) and the initial
hypersurface $\SigmaEY=\Set{y^0=y^0_0}$ one has
$\tau_0(x,v)=({y^0_0-x^0})/{v^0}$
and the map
$\varphi$ is given by (\ref{SusKernel_def_varphi}) with
$\phi_{\tau}^\star(y^a)=x^a+\tau y^a$ and
$\phi_{\tau}^\star(u^a)=v^a$
From (\ref{Plasma_susep}) with $\chi$ given by (\ref{Plasma_Eg_res})
and $Z_1[\zetabd_1]$ given by (\ref{SusKernel_Z1_soln}) one has:
\begin{equation}
\begin{aligned}
\lefteqn{\Pi_1[F_1,\zetabd_1]=}\quad&
\\
&
-\sum_\speciesa \frac{q^{\speciesa\,2}}{m^\speciesa} E e^{-i\omega x^0+i\kk x^1}
\bigg\{
 dx^{01} \int dv^{123}  T^\speciesa \frac{(v^0)^2 -(v^1)^2}{v^0}
+dx^{12} \int dv^{123}  T^\speciesa v^2
\\&
-dx^{02} \int dv^{123}  T^\speciesa \frac{v^2v^1}{v^0}
+dx^{13} \int dv^{123}  T^\speciesa v^3
+dx^{03} \int dv^{123}  T^\speciesa \frac{v^3v^1}{v^0}
\bigg\}
\\&
+\sum_\speciesa {q^{\speciesa}}\bigg\{
dx^{02}\int dv^{123}
\Big(\zetao{4}\frac{v^1(x^0-y_0^0)}{v^0}-\zetao{1}{v^1}\Big)
+
dx^{03}\int dv^{123}
\Big(\zetao{5}\frac{v^1(x^0-y_0^0)}{v^0}-\zetao{2}{v^1}\Big)
\\&
+
dx^{12}\int dv^{123} \Big(v^0\zetao{1}-\zetao{4}(x^0-y_0^0)\Big) +
dx^{13}\int dv^{123} \Big(v^0\zetao{2}-\zetao{5}(x^0-y_0^0)\Big)
\\&
+
dx^{23}\int dv^{123} \bigg(
\zetao{3}
+\zetao{4}\frac{v^1 v^3 (x^0-y_0^0)}{(v^0)^2}
-\zetao{5}\frac{v^1 v^2 (x^0-y_0^0)}{(v^0)^2}
+\zetao{6}(x^0-y_0^0)\Big(\frac{v^1}{v^0}-1\Big)
\bigg)\bigg\}
\\&
+\star d\big( \Xi_1[F_1]\big)
+
\star d\big( \Zonez[\zetabd_1]\big)
\end{aligned}
\label{Plasma_neut_Pi1}
\end{equation}
where $\int dv^{123}$ denotes the triple integral operator
$\iiint_{-\infty}^\infty dv^{123}$, $v^0=\sqrt{1+v_\mu v^\mu}$,
\begin{align}
T^\speciesa=T^\speciesa(x,v)=\int_{(y^0_0-x^0)/v^0}^0
e^{i\tau(-\omega v^0+\kk v^1)}
{f_0^\speciesa(x+\tau v,v)} \tau
d\tau
\label{Plasma_neut_T}
\end{align}
and
$\zetao{A}=\zetao{A}(x^\mu,v^\mu)=\zetaoh{A}\big(x^\mu-{x^0v^\mu}/{v^0},v^\nu\big)$
in (\ref{Plasma_neut_Pi1}). This response is not in general plane
fronted.

For the particular case of a plane fronted plasma distribution:
\begin{align}
f_0^\speciesa(x,v)=h_0^\speciesa(x^0,x^1,v^1)\delta(v^2)\delta(v^3)
\label{Plasma_new_f0}
\end{align}
with initial data:
\begin{align*}
\zeta^\speciesa_1=0
\end{align*}
(\ref{Plasma_neut_Pi1})
becomes the plane fronted 2-form
\begin{equation}
\begin{aligned}
\lefteqn{\Pi_1[F_1,\zetabd]|_x}&
\\
&=-
dx^{01} \sum_\speciesa \frac{q^{\speciesa\,2}}{m^\speciesa} E e^{-i\omega x^0+i\kk x^1}
\int_{-\infty}^\infty \!\!  dv^1
\int_{(y^0_0-x^0)/v^0}^0 \!\!\!\! d\tau\,
e^{i\tau(-\omega v^0+\kk v^1)}
{h_0^\speciesa(x^0+\tau v^0,x^1+\tau v^1,v^1)}
\frac{\tau}{v^0}
\\&\qquad
+
\star d\big( \Xi_1[F_1]\big)
\end{aligned}
\label{Plasma_neut_T01}
\end{equation}
describing the response of a spacetime inhomogeneous unbounded plasma to
(\ref{Plasma_plane_wave}).

\subsection{Spacetime homogeneous unbounded plasmas}

The previous discussion
simplifies considerably if the unperturbed plasmas is homogeneous in
space and time.  In Minkowski spacetime $M$, an unbounded unperturbed plasma is
deemed \defn{spacetime homogeneous} if $A_z^\star F_0=F_0$ and
$\Adot_z^\star\theta^\speciesa_0=\theta^\speciesa_0$ for all $z\in M$
where the translation map $A_z:M\to M$, $A_z(x)=x+z$ induces the map
$\Adot_z:\Ebun\to\Ebun$, $\Adot_z=A_{z\star}$.  Such
spacetime homogeneity implies that in all inertial frames the medium
is stationary and spatially homogeneous in all directions.  Such a
spacetime homogeneous plasma will give rise to a spacetime homogeneous
electromagnetic constitutive relation.  In addition to the components
$(F_0)_{ab}$ with respect to an inertial frame being constant, the
functions $f^\speciesa(x,v)$ are independent of event position $x$ and can
therefore be written $f^\speciesa(v)$.

In this scenario the Fourier transform (\ref{Disp_homo_chi_k}) of the
susceptibility kernel (\ref{Disp_homo_freq}) for each species,
is then given by
\begin{equation}
\begin{aligned}
&\FT\chi^\speciesa{}_{ab}{}^{ef}(k) dx^{ab}
\\&\quad
=
\tfrac{1}{2}q^\speciesa dx_{gh}
\int_{-\infty}^{0} d\tau \int dv^{123}
f^\speciesa_0(v) e^{-i k\cdot\RL v}
\frac{v^g}{v_0}
\big(g^{\inu e} u^f -g^{\inu f} u^e\big)
\Big(\IL_\nu^h(\tau)-\frac{u_\nu}{u_0}
\IL_0^h(\tau)\Big)
\end{aligned}
\label{SusKernel_homo_Derfler}
\end{equation}
where $\MF_0$ is the $4\times4$ real matrix with components
$(\MF_0)^a{}_{b}=\eta^{ac} (F_{0})_{cb}$ generating the matrices
\begin{equation}
\begin{aligned}
&\MD^a_b(\tau)=\exp\Big(\tau \frac{q^\speciesa}{m^\speciesa}
  \MF_0\Big)^a_{\ b}
\,,\qquad
\ID_b^a(\tau)=g_{bc} \MD^c_d(\tau) g^{da}
\,,\\
&\ML^a_b(\tau)=\int^{\tau}_0 \MD^a{}_b(\tau')\, d\tau'
\,,\qquad
\IL_b^a(\tau)=g_{bc} \ML^c_d(\tau) g^{da}\,,
\end{aligned}
\label{SusKernel_homo_def_DL}
\end{equation}
$k\cdot\RL v=k_a\ML^a_b(\tau) v^b$ and
\begin{align}
u^a(\tau,v^1,v^2,v^3) = \MD^a_b(\tau) v^b
\label{SusKernel_homo_phi_tau_u}
\end{align}
The susceptibility kernel (\ref{SusKernel_homo_Derfler}) can be shown
to agree with
the results of O'Sullivan and
Derfler \cite{o1973relativistic}.

Furthermore for a microscopically neutral spacetime homogeneous plasma
with $F_0=0$, $G_1=0$ and
$f^\speciesa_0(v)=h^\speciesa_0(v^1)\delta(v^2)\delta(v^3)$ it follows
from (\ref{Plasma_neut_T01}) and (\ref{Plasma_MV_G1}) that for
$\Im(\omega)>0$
\begin{align}
1=
\sum_\speciesa\frac{q^{\speciesa\,2}}{m^\speciesa\epsilon_0}
\int_{-\infty}^{\infty}
\frac{{h^\speciesa_0(v^1)\,} dv^{1}}{v^0(-\omega v^0+\kk v^1)^2}
\label{Plasma_disp_homo}
\end{align}
The relativistic Landau damped dispersion
relation for plane fronted Langmuir modes in an unperturbed
spacetime homogeneous plasma arises by analytic continuation of the
integral (\ref{Plasma_disp_homo}) to the lower-half  complex $\omega$ plane.

\subsection{Langmuir modes for an  inhomogeneous unbounded
  plasma in Minkowski spacetime}

If the plasma is microscopically neutral but spacetime inhomogeneous
in its unperturbed state the Landau dispersion relation corresponding
to (\ref{Plasma_disp_homo}) becomes more involved. We define the
generalized Langmuir sector to contain perturbations described by
(\ref{Plasma_neut_T01}) but with the external polarization specified
by $\Xi_1[F_1]$ set to zero. Since $\zeta^\speciesa_1=0$,
$\Pi_1[F_1,0]$ will be denoted $\Pi_1[F_1]$. Thus (\ref{Plasma_MV_G1})
with $G_1=0$ becomes
\begin{align}
\epsilon_0 F_1=-\Pi_1[F_1]
\label{Plasma_disp_restrict}
\end{align}
Consider the case where planar inhomogeneities in a plasma composed of
electrons and ions arise from the unperturbed spacetime inhomogeneous
solution to the Maxwell-Vlasov system:
(\ref{Plasma_MV0}-\ref{Plasma_MV_W0}) with $F_0=0$ and
\begin{equation}
\begin{aligned}
f_0^\speciesel(x^0,x^1,x^2,x^3,v^1,v^2,v^3)
&=
f_0^\speciesion(x^0,x^1,x^2,x^3,v^1,v^2,v^3)
\\&=
h\Big(x^1-\frac{v^1 x^0}{v^0},v^1\Big) \delta(v^2) \delta(v^3)\
\end{aligned}
\label{eg_dist}
\end{equation}
where $q^\speciesel=-q^\speciesion$.

For example one might consider
\begin{align*}
h(x^1,v^1) = n^\speciesion(x^1) A^\speciesion(x^1)
\exp\Big( - \frac{m^\speciesion{v^0}}{k_B T^\speciesion(x^1)} \Big)
\end{align*}
where $A^\speciesion(x^1)$ normalizes (\ref{eg_dist}). Then
$f^\speciesion$ initially at $x^0=0$ represents a distribution of ions
where, at each spatial point $x^1$, the velocities belong to the
1-dimensional Maxwell-J\"uttner distribution.
In such a distribution the
temperature $T^\speciesion(x^1)$ and the number density of ions
$n^\speciesion(x^1)$ depend on position. It follows from
(\ref{eg_dist}) that $f^\speciesel$ also initially represents a
position dependent Maxwell-J\"uttner distribution where
$n^\speciesel(x^1)=n^\speciesion(x^1)$ and
$T^\speciesel(x^1)=T^\speciesion(x^1)m^\speciesel/m^\speciesion$. After
the initial moment, the ions and electrons drift according to
(\ref{eg_dist}) and velocities do not remain in the Maxwell-J\"uttner
distributions.  Alternatively (\ref{eg_dist}) might describe a plasma
composed of particles and anti-particles.

In the theory of a spacetime homogeneous plasma $\omega$ and $\kk$
satisfy the transcendental dispersion relation
(\ref{Plasma_disp_homo}). This relation contains an integral that is
potentially singular. The Landau prescription circumvents this
singularity by complexifying $\omega$ and defining an analytic
continuation for the integral in the complex $\omega$ plane.

Setting $h_0^\speciesa(x^0,x^1,v^1)=h\big(x^1-{v^1
  x^0}/{v^0},v^1\big)$ in (\ref{Plasma_new_f0}) yields
(\ref{eg_dist}) and (\ref{Plasma_neut_T01}) becomes
\begin{equation}
\begin{aligned}
\lefteqn{\Pi_1[F_1]|_x}\ &
\\
&=-
dx^{01} q^{\speciesel\,2}
\Big(\frac{1}{m^\speciesion}+\frac{1}{m^\speciesel}\Big)
E e^{-i\omega x^0+i\kk x^1}
\!\!\int_{-\infty}^\infty  dv^1
{h\Big(x^1-\frac{v^1x^0}{v^0},v^1\Big)}
\!\!\int_{(y^0_0-x^0)/v^0}^0 \!\!\!\!\!\!d\tau\,
e^{i\tau(-\omega v^0+\kk v^1)}
\frac{\tau}{v^0}
\end{aligned}
\label{Plasma_Lan_Pi_pre}
\end{equation}
To compare with the results (\ref{Plasma_disp_homo}) given for the
homogeneous case, consider the limit $y^0_0\to-\infty$ with $\Im(\omega)>0$.
Furthermore for the non-evanescent modes considered here
$\Im(\kk)=0$. Thus (\ref{Plasma_Lan_Pi_pre}) becomes
\begin{equation}
\begin{aligned}
\Pi_1[F_1]|_x=
-dx^{01} \epsilon_0\kappaQM
E e^{-i\omega x^0+i\kk x^1}
\int_{-\infty}^\infty  dv^1
\frac{h(x^1-{v^1x^0}/{v^0},v^1)}
{v^0 (-\omega v^0+\kk v^1)^2}
\end{aligned}
\label{Plasma_Lan_Pi}
\end{equation}
where
\begin{align*}
\kappaQM=\frac{q^{\speciesel\,2}}{\epsilon_0 m^\speciesion}+
\frac{q^{\speciesel\,2}}{\epsilon_0 m^\speciesel}
\end{align*}
In a
spacetime inhomogeneous plasma there is no time-harmonic solution
or associated transcendental dispersion relation between $\omega$ and
$\kk$.  We therefore propose solving (\ref{Plasma_disp_restrict})
with a longitudinal field $F_1$ represented as the packet
\begin{align}
F_1{(x^0,x^1)}=dx^{01} \int_{-\infty}^\infty d\omegaHat
\int_{-\infty}^\infty d \kkHat\
\FT{E}(\omegaHat,\kkHat) e^{-i\omegaHat x^0 + i \kkHat x^1}
\label{Plasma_Lan_F1}
\end{align}
Substituting (\ref{Plasma_Lan_Pi}) and (\ref{Plasma_Lan_F1}) into
(\ref{Plasma_disp_restrict}) yields
\begin{align*}
&\int_{-\infty}^\infty d\omegaHat
\int_{-\infty}^\infty d \kkHat\
\FT{E}(\omegaHat,\kkHat) e^{-i\omegaHat x^0 + i \kkHat x^1}
\\&\qquad\qquad
=
\kappaQM
\int_{-\infty}^\infty d\omegaHat
\int_{-\infty}^\infty d \kkHat\
\FT{E}(\omegaHat,\kkHat) e^{-i\omegaHat x^0+i\kkHat x^1}
\int_{-\infty}^\infty  dv^1
\frac{h(x^1-{v^1x^0}/{v^0},v^1)}
{v^0 (\omegaHat v^0+\kkHat v^1)^2}
\end{align*}
Performing the inverse Fourier transform gives
\begin{align*}
&4\pi^2\FT{E}(\omega,\kk)
\\&=
\kappaQM
\int_{-\infty}^\infty d x^0
\int_{-\infty}^\infty d x^1
\int_{-\infty}^\infty d\omegaHat
\int_{-\infty}^\infty d \kkHat\
\FT{E}(\omegaHat,\kkHat) e^{i(-(\omegaHat-\omega) x^0+(\kkHat-\kk) x^1)}
\int_{-\infty}^\infty  dv^1
\frac{h(x^1-{v^1x^0}/{v^0},v^1)}
{v^0 (\omegaHat v^0+\kkHat v^1)^2}
\end{align*}
Since
\begin{align*}
\int_{-\infty}^\infty d x^0
\int_{-\infty}^\infty d x^1
e^{i(-(\omegaHat-\omega) x^0+(\kkHat-\kk) x^1)}
{h(x^1-{v^1x^0}/{v^0},v^1)}
=
2\pi \FT{h}(\kk-\kkHat,v^1)
\delta\big(\omegaHat-\omega+v^1(\kk-\kkHat)/v^0\big)
\end{align*}
where
\begin{align*}
\FT{h}(\kk,v^1) =
\int_{-\infty}^\infty e^{-i \kk s} h(s,v^1) ds
\end{align*}
one has
\begin{align}
&\FT{E}(\omega,\kk)=
\frac{\kappaQM}{2\pi}
\int_{-\infty}^\infty d\omegaHat
\int_{-\infty}^\infty d \kkHat\
\int_{-\infty}^\infty  dv^1
\frac{\FT{E}(\omegaHat,\kkHat) }
{v^0 (\omegaHat v^0+\kkHat v^1)^2}
\FT{h}(\kk-\kkHat,v^1)
\delta\big(\omegaHat-\omega+v^1(\kk-\kkHat)/v^0\big)
\label{Plasma_Ehat1}
\end{align}
Since we restrict to non-evanescent modes $\kk$ and $\kkHat$ are real.
For $\FT{E}(\omega,\kk)$ to be non-zero one requires the argument of
the $\delta$-function to be zero. Since $v^1$ is real and therefore
$v^1(\kk-\kkHat)/v^0$ is real it follows that although $\Im(\omega)>0$
and $\Im(\omegaHat)>0$ the difference $\omega-\omegaHat$ is
real. Furthermore from $\omegaHat-\omega+v^1(\kk-\kkHat)/v^0=0$ it
follows that $|\omegaHat-\omega|<|\kkHat-\kk|$. Thus
(\ref{Plasma_Ehat1}) becomes
\begin{align}
&\FT{E}(\omega,\kk)=
\frac{\kappaQM}{2\pi}
\int_{-\infty}^\infty d \kkHat\
\II(\omega,\kk,\kkHat)
\label{Plasma_Ehat2}
\end{align}
where
\begin{align}
\II(\omega,\kk,\kkHat)=
\int_{S(\omega,k,\kHat)}
d\omegaHat\
\FT{E}(\omegaHat,\kkHat) \frac{(\kk-\kkHat)}
{(\omegaHat \kk-\kkHat \omega)^2}
\FT{h}\bigg(\kk-\kkHat,
\frac{k-\kkHat}{\sqrt{(\kkHat-\kk)^2-(\omegaHat-\omega)^2}}\bigg)
\label{Plasma_II}
\end{align}
and the contour of integration for $\omegaHat$ in (\ref{Plasma_II}) is
the straight line $S(\omega,k,\kHat)$ where $\Im(\omegaHat)=\Im(\omega)>0$ and
$-|\kkHat-\kk|<\Re(\omegaHat-\omega)<|\kkHat-\kk|$. Since
$(\omegaHat-\omega)^2<(\kkHat-\kk)^2$ the arguments of $\FT{h}$ in
(\ref{Plasma_II}) are always real and non-singular on $S(\omega,k,\kHat)$.

To accommodate the situation when $\FT{E}(\omega,\kk)$ describes
damped electromagnetic waves one must continue (\ref{Plasma_II}) to
$\Im(\omega)<0$ for real $\kk$. However there is a double pole in the
complex $\omegaHat$ plane at
$\omegaHat=\omegaHat_0=\kkHat\omega/\kk$ that coincides with
$S(\omega,k,\kHat)$ when $\Im(\omega)=0$ and $|\omega|<|\kk|$. To
define an analytic continuation of (\ref{Plasma_II}) to $\Im(\omega)<0$
when $|\Re(\omega)|<|\kk|$, we indent $S(\omega,k,\kHat)$ to encircle
the pole in the standard manner and write the contour integral in
terms of a principle part and associated residue, see figure
\ref{fig_omegahat_contour}. Such a continuation scheme gives rise to
branches in the $\omega$ plane for $\II(\omega,\kk,\kkHat)$ as shown
in figure \ref{fig_branch_cuts_I}.

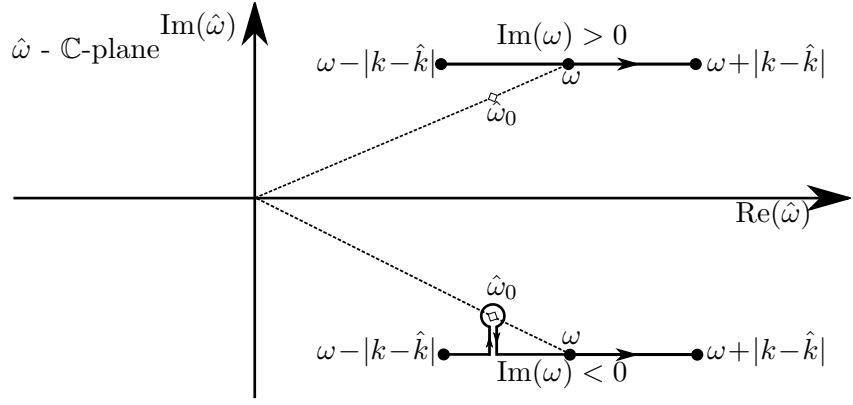
\begin{figure}
\begin{center}
\setlength{\unitlength}{0.04\textwidth}
\begin{picture}(16,8)
\put(0,0){\includegraphics[height=8\unitlength, viewport=0 1 501 240]{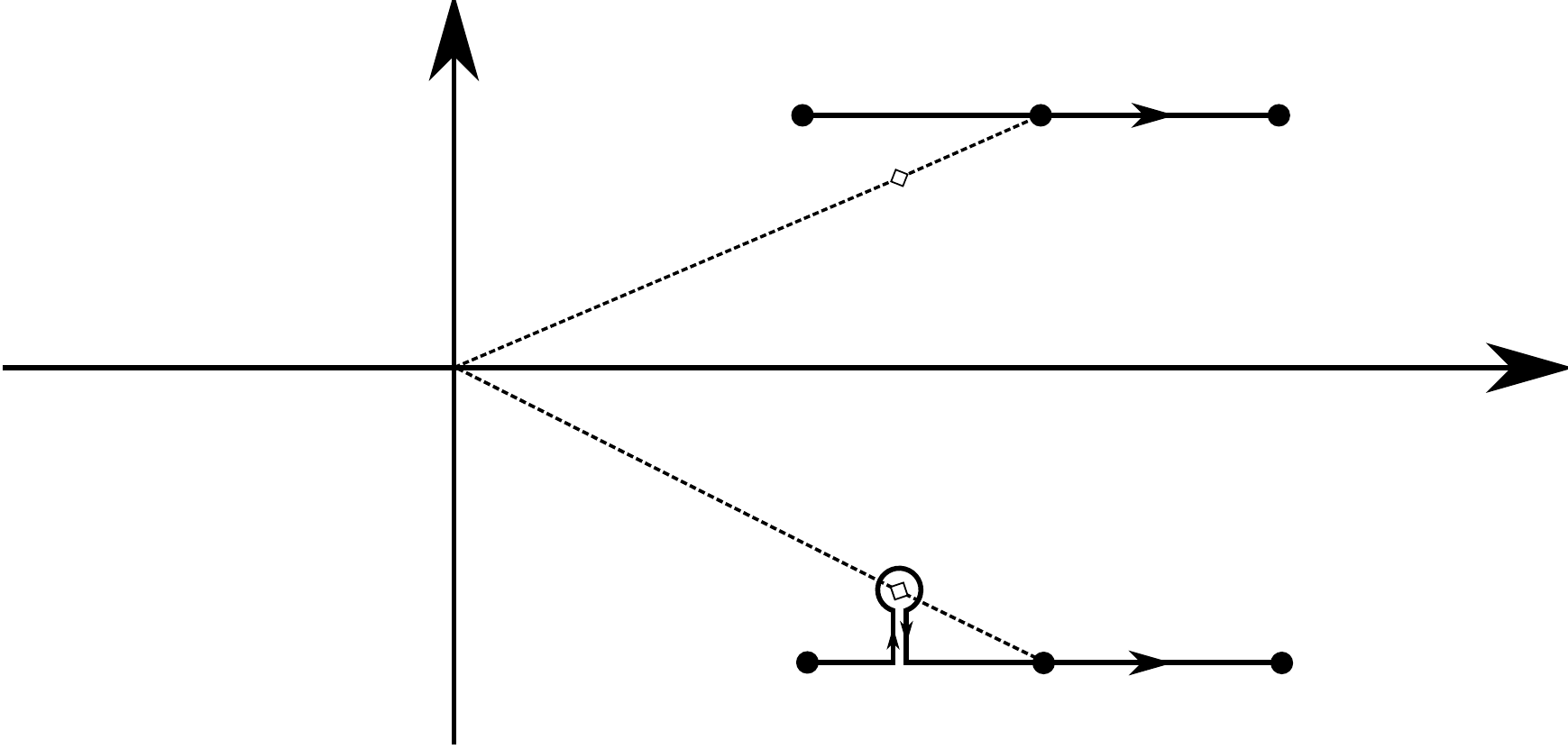}}
\put(0,7){\makebox(0,0)[l]{$\omegaHat$ - $\Cmpx$-plane}}
\put(4.5,7.5){\makebox(0,0)[r]{$\Im(\omegaHat)$}}
\put(16,4){\makebox(0,0)[tr]{$\Re(\omegaHat)$}}
\put(11,7){\makebox(0,0)[bc]{$\Im(\omega)>0$}}
\put(11,0.8){\makebox(0,0)[tc]{$\Im(\omega)<0$}}
\put(11,6.6){\makebox(0,0)[tl]{\small $\omega$}}
\put(11,1.1){\makebox(0,0)[bl]{\small $\omega$}}
\put(9.5,6){\makebox(0,0)[tl]{\small $\omegaHat_0$}}
\put(9.5,2){\makebox(0,0)[bl]{\small $\omegaHat_0$}}
\put(13.9,1){\makebox(0,0)[l]{\small $\omega\!+\!|k\!-\!\kHat|$}}
\put(8.5,1){\makebox(0,0)[r]{\small $\omega\!-\!|k\!-\!\kHat|$}}
\put(13.9,6.8){\makebox(0,0)[l]{\small $\omega\!+\!|k\!-\!\kHat|$}}
\put(8.5,6.8){\makebox(0,0)[r]{\small $\omega\!-\!|k\!-\!\kHat|$}}
\end{picture}
\end{center}
\caption{The upper contour denotes $S(\omega,k,\kHat)$ when
  $\Im(\omega)>0$ for real $k,\kHat$.
The lower contour of integration is used when $\Im(\omega)<0$ for real
$k,\kHat$.}
\label{fig_omegahat_contour}
\end{figure}

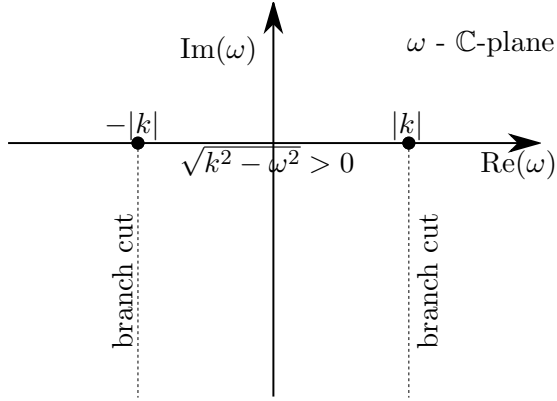
\begin{figure}
\begin{center}
\setlength{\unitlength}{0.04\textwidth}
\begin{picture}(10,8)
\put(0,0){\includegraphics[height=8\unitlength, viewport=0 0 432 324]{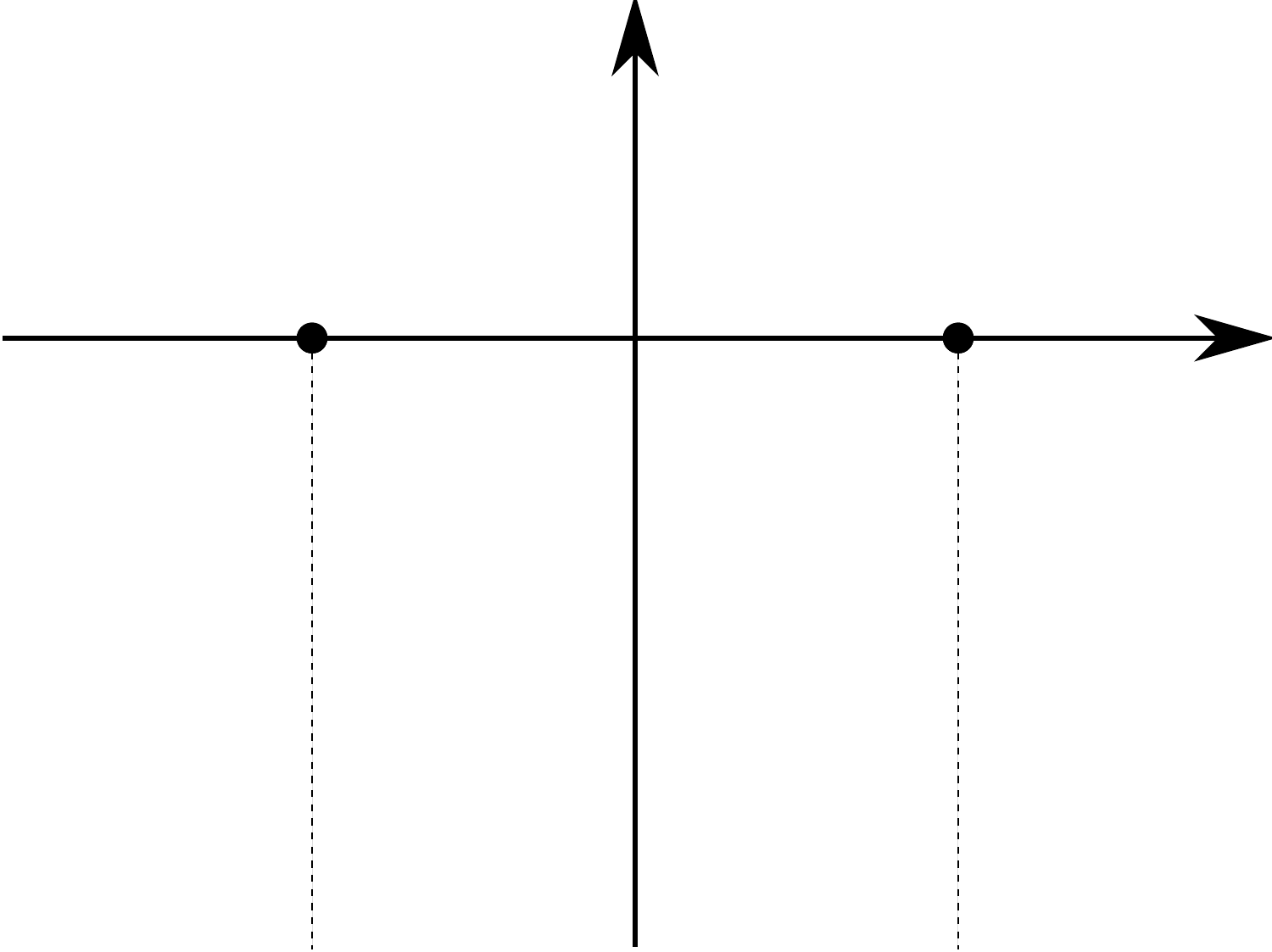}}
\put(8,7){$\omega$ - $\Cmpx$-plane}
\put(5,7){\makebox(0,0)[r]{$\Im(\omega)$}}
\put(11,5){\makebox(0,0)[tr]{$\Re(\omega)$}}
\put(8,5.2){\makebox(0,0)[b]{$|\kk|$}}
\put(2.5,5.2){\makebox(0,0)[b]{$-|\kk|$}}
\put(2.1,1){\rotatebox{90}{\small branch cut}}
\put(8.2,1){\rotatebox{90}{\small branch cut}}
\put(5.2,5.1){\makebox(0,0)[ct]{$\sqrt{\kk^2-\omega^2}>0$}}
\end{picture}
\end{center}
\caption{Branch cuts in $\omega$ for $\II(\omega,\kk,\kkHat)$.}
\label{fig_branch_cuts_I}
\end{figure}

This analytic continuation of (\ref{Plasma_II}) to
$\Im(\omega)<0$  aquires the residue
\begin{align*}
R(\omega,k,\kHat)
&=
\frac{|k-\kHat|}{k\,|k|}
\pfrac{\FT{E}}{\omega}\Big(\frac{\omega\kHat}{k},\kHat\Big)
\FT{h}\Big(k-\kHat,\frac{s_k s_{k-\kHat} \omega}{\sqrt{k^2-\omega^2}}\Big)
\\& \quad
-
\frac{k}{(k^2-\omega^2)^{3/2}}
\FT{E}\Big(\frac{\omega\kHat}{k},\kHat\Big)
\,\hHat_{v^1}\!
\Big(k-\kHat,\frac{s_k s_{k-\kHat} \omega}{\sqrt{k^2-\omega^2}}\Big)
\end{align*}
where $\hHat_{v^1}(k,v^1)=\pfrac{\hHat}{v^1}(k,v^1)$, $s_k=k/|k|$ and
$s_{k-\kHat}=(k-\kHat)/|k-\kHat|$. In the case when $\Im(\omega)=0$,
the principle value of (\ref{Plasma_II}) is taken together with
residue $\tfrac12 R(\omega,k,\kHat)$. Equation (\ref{Plasma_Ehat2})
then gives
\begin{equation}
\begin{aligned}
\FT{E}(\omega,k)
&=
\frac{\kk}{2\pi}
\int_{-\infty}^\infty
I(\omega,k,\kHat)\,d\kHat
\quadtext{if} \Im(\omega) > 0
\quadtext{or} |\Re(\omega)|>|k|
\\
\FT{E}(\omega,k)
&=
\frac{\kk}{2\pi}
\int_{-\infty}^\infty
I(\omega,k,\kHat)\,d\kHat
- i\kk
\int_{-\infty}^{\infty}
 R(\omega,k,\kHat)\,d\kHat
\\&
\hspace{20em}\quadtext{if} \Im(\omega) < 0\quadtext{and}|\Re(\omega)|\le|k|
\\
\FT{E}(\omega,k)
&=
\frac{\kk}{2\pi}
\int_{-\infty}^\infty
{\cal P}I(\omega,k,\kHat)\,d\kHat
-\frac{ i\kk}{2}
\int_{-\infty}^{\infty}
R(\omega,k,\kHat)\,d\kHat
\\&
\hspace{20em}\quadtext{if} \Im(\omega) = 0\quadtext{and}|\Re(\omega)|<|k|
\end{aligned}
\label{Plasma_III}
\end{equation}
where ${\cal P}I(\omega,k,\kHat)$ in (\ref{Plasma_III}) refers to the
principle part of (\ref{Plasma_II}) when $\Im(\omega) = 0$ and
$|\Re(\omega)|<|k|$ and hence the pole at $\omegaHat_0$ lies on the
contour $S(\omega,k,\kHat)$.
Thus in each domain above, the perturbation $\hat{E}(\omega,k)$ must
be determined by solving a non-standard integral equation.

\section{Conclusions}

In this article a classical covariant description of electromagnetic
interactions in continuous matter in an arbitrary background
gravitational field has been formulated in terms of a polarization
2-form that enters into the macroscopic Maxwell equations. Linear
dispersive constitutive relations arise when this 2-form is expressed
as an affine functional of the Maxwell 2-form with the aid of a
2-point susceptibility kernel.  We have explored the constraints on
this kernel imposed by causality requirements, spacetime Killing
symmetries and local gauge freedoms.  The formalism has been applied
to an analysis of constitutive models for waves in collisionless
plasmas. In particular a formula for the linear susceptibility of a
fully ionized inhomogeneous unbounded non-stationary collisionless
plasma to a perturbation in the presence of gravity has been given in
terms of maps describing the dynamics of the plasma. This formula has
been elucidated by reference to both homogeneous and inhomogeneous
perturbations in Minkowski spacetime. In the former case one recovers
the standard Landau dispersion relation when perturbing Langmuir
modes. In the latter case we have described a generalized damping
mechanism for such modes that may arise when the unperturbed state is
both inhomogeneous and non-stationary. Such a mechanism arises from
the analytic continuation of an integral equation that replaces the
Landau dispersion relation.

It is concluded that the use of a covariant 2-point affine
susceptibility kernel in describing the electromagnetic response of
dispersive media offers a modelling tool that naturally generalizes
the use of permittivity and permeability tensors used to model
electromagnetic interactions in non-relativistic media. The
formulation in terms of an arbitrary background spacetime metric
offers potential applications in a number of astrophysical contexts
involving electromagnetic fields in inhomogeneous or non-stationary
plasmas


\section*{Acknowledgements}
The authors are grateful to support from EPSRC (EP/E001831/1) and the
Cockcroft Institute (STFC ST/G008248/1). 

\bibliographystyle{unsrt}

\appendix
\section{Proofs of results used used in the text.}

\begin{lemma}
\label{lm_fibre_int}
Local representation of $\mapints_{\pi_\Nbun}\alpha$
in \textup{(\ref{Disp_def_fibre_int})} from the implicit definition in
equation \textup{(\ref{Notation_mapint})}.
\end{lemma}

\begin{proof}
On a fibred manifold $\Nbun$ of dimension $n+r$ with projection
$\pi_\Nbun:\Nbun\to N$ over a manifold $N$ of dimension $n$. Thus at
each point $\sigma\in N$ one has the fibre
$\Nbun_\sigma=\pi_\Nbun^{-1}\Set{\sigma}=
\Set{(\sigma',\varsigma)\in\Nbun\,\big|\,
  \pi_\Nbun(\sigma',\varsigma)=\sigma}$ so $\dim(\Nbun_\sigma)=r$ is
the fibre dimension.  Let $(\sigma^1,\ldots,\sigma^n)$ and
$(\sigma^1,\ldots,\sigma^n,\varsigma^1\ldots \varsigma^r)$ be local
coordinates for patches on $N$ and $\Nbun$ respectively.

Consider first the case when $\alpha\in\Gamma\Lambda^{p+r}\Nbun$
consists of a single component $\alpha_I(\sigma,\varsigma)
d\sigma^I\wedge d\varsigma^{1\ldots r}$ with no sum on $I$. Hence
explicit summation will be used in this particular proof. Set
${\hat{I}}=\Set{1,\ldots,n}\backslash I$ so that
$d\sigma^{\hat{I}}\wedge d\sigma^I=\pm d\sigma^{1\ldots n}$ and let
$\beta=\sum_J\beta_{J}(\sigma) d\sigma^{J}$ then $\beta\wedge
d\sigma^I=\pm \beta_{\hat{I}}\alpha_I d\sigma^{1\ldots n}$ so that:
\begin{align*}
\lefteqn{
\sum_J
\int_{(\sigma,\varsigma)\in\Nbun} \pi_\Nbun^\star(\beta_{J}(\sigma) d\sigma^{J})
\wedge \alpha_I(\sigma,\varsigma) d\sigma^I\wedge
d\varsigma^{1\ldots r}}\qquad\qquad&
\\
&=
\sum_J
\int_{(\sigma,\varsigma)\in\Nbun} \beta_{J}(\sigma) d\sigma^{J}
\wedge \alpha_I(\sigma,\varsigma) d\sigma^I\wedge
d\varsigma^{1\ldots r}
\\&=
\sum_J
\int_{(\sigma,\varsigma)\in\Nbun} \beta_{J}(\sigma) d\sigma^{J}
\wedge  d\sigma^I\wedge
\alpha_I(\sigma,\varsigma) d\varsigma^{1\ldots r}
\\&=
\int_{(\sigma,\varsigma)\in\Nbun} \beta_{\hat{I}}(\sigma) d\sigma^{\hat{I}}
\wedge  d\sigma^I\wedge
\alpha_I(\sigma,\varsigma) d\varsigma^{1\ldots r}
\\&=
\int_{\sigma\in N} \beta_{\hat{I}}(\sigma) d\sigma^{\hat{I}}\wedge d\sigma^I
\int_{\Nbun_\sigma}\alpha_I(\sigma,\varsigma)
d\varsigma^{1\ldots r}
\\&=
\sum_J
\int_{\sigma\in N} \beta_{J}(\sigma) d\sigma^{J}\wedge d\sigma^I
\int_{\Nbun_\sigma}\alpha_I(\sigma,\varsigma)
d\varsigma^{1\ldots r}
\\&=
\int_{\sigma\in N} \beta \wedge d\sigma^I
\int_{\Nbun_\sigma}\alpha_I(\sigma,\varsigma)
d\varsigma^{1\ldots r}
\end{align*}
Thus by linearity
\begin{align}
\int_{\Nbun}\pi_\Nbun^\star(\beta) \wedge
\alpha
&=
\sum_I \int_{\sigma\in N} \beta \wedge d\sigma^I
\int_{\Nbun_\sigma}\alpha_I(\sigma,\varsigma)
d\varsigma^{1\ldots r}
\label{Proofs_mapint}
\end{align}
where $\alpha=\sum_I \alpha_I(\sigma,\varsigma)d\sigma^I\wedge
d\varsigma^{1\ldots r}$.  If (\ref{Disp_def_fibre_int}) holds then for
$\alpha=\sum_I \alpha_I(\sigma,\varsigma)d\sigma^I\wedge
d\varsigma^{1\ldots r}$,
\begin{align*}
\int_N\beta \wedge
\mapint_{\pi_\Nbun}\alpha
&=
\sum_I
\int_N\beta \wedge d\sigma^I
\int_{\varsigma\in \Nbun_\sigma}
\IC{\sigma}{I}\alpha|_{(\sigma,\varsigma)}
=
\sum_I
\int_N\beta \wedge d\sigma^I
\int_{\varsigma\in \Nbun_\sigma}
\alpha_Id\varsigma^{1\ldots r}
\\&=
\int_{\Nbun}\pi_\Nbun^\star(\beta) \wedge
\alpha
\end{align*}
Hence (3). Conversely if (\ref{Notation_mapint}) holds for
$\alpha=\sum_I \alpha_I(\sigma,\varsigma)d\sigma^I\wedge
d\varsigma^{1\ldots r}$ then from (\ref{Proofs_mapint})
\begin{align*}
\int_N\beta \wedge
\mapint_{\pi_\Nbun}\alpha
=
\int_{\Nbun}\pi_\Nbun^\star(\beta) \wedge
\alpha
=
\sum_I
\int_N\beta \wedge d\sigma^I
\int_{\varsigma\in \Nbun_\sigma}
\IC{\sigma}{I}\alpha|_{(\sigma,\varsigma)}
\end{align*}
Since this is true for all $\beta$ then (\ref{Disp_def_fibre_int})
holds.

If $\alpha$ does not contain the factor $\varsigma^{1\ldots r}$
i.e. $\alpha= \alpha_{IK}(\sigma,\varsigma)d\sigma^I\wedge
d\varsigma^K$ where $K\ne \Set{1,\ldots,r}$ then the right hand side
of (\ref{Notation_mapint}) becomes
\begin{align*}
\int_\Nbun \pi_\Nbun^\star(\beta)\wedge \alpha
=
\sum_J \int_\Nbun \beta_J \alpha_{IK}(\sigma,\varsigma)
d\sigma^J \wedge d\sigma^I \wedge d\varsigma^K
= 0
\end{align*}
and the right hand side of (\ref{Disp_def_fibre_int}) becomes
\begin{align*}
\sum_{I}
d\sigma^I
\int_{\varsigma\in \Nbun_\sigma}
\alpha_{IK} (\sigma,\varsigma)
d\varsigma^K = 0
\end{align*}
Thus by linearity (\ref{Notation_mapint}) and
(\ref{Disp_def_fibre_int}) are equivalent for all $\alpha$.
\end{proof}

\begin{lemma}
Verification of equation \textup{(\ref{Notation_mapint_comm_d})}:
\begin{align*}
\bigg(d \mapint_{\pi_\Nbun} \alpha\bigg)\bigg|_{\sigma}
= \bigg(\mapint_{\pi_\Nbun} d \alpha\bigg)\bigg|_{\sigma}
\end{align*}
\label{lm_d_mapint}
\end{lemma}

\begin{proof}
Let $\deg(\alpha)=p+r$, $\deg(\beta)=n-p-1$ and $\partial N$ and
$\partial\Nbun$ be the boundaries of $N$ and $\Nbun$. Since
$\sigma\notin\partial N$ one may choose $\beta$ to have support away
from $\partial N$ thus
\begin{align*}
\int_{\partial N}\beta\wedge\Big(\mapint_{\pi_\Nbun}\alpha\Big)=0
\end{align*}
and since $\alpha$ has support away from $\partial\Nbun$ then
\begin{align*}
\int_{\partial \Nbun}\pi_\Nbun^\star\beta\wedge\alpha=0
\end{align*}
It follows that
\begin{align*}
\int_N \beta\wedge\Big(\mapint_{\pi_\Nbun}d\alpha\Big)
&=
\int_\Nbun \pi_\Nbun^\star(\beta)\wedge d\alpha
\\&=
(-1)^{n-p-1}\int_\Nbun d\big(\pi_\Nbun^\star(\beta)\wedge \alpha\big)
+(-1)^{n-p}\int_\Nbun d\pi_\Nbun^\star(\beta)\wedge\alpha
\\&=
(-1)^{n-p-1}\int_{\partial\Nbun} \pi_\Nbun^\star(\beta)\wedge \alpha
+(-1)^{n-p}\int_\Nbun \pi_\Nbun^\star(d\beta)\wedge\alpha
\\&=
(-1)^{n-p}\int_N d\beta\wedge\Big(\mapint_{\pi_\Nbun}\alpha\Big)
\\&=
(-1)^{n-p}\int_N d\Big(\beta\wedge\Big(\mapint_{\pi_\Nbun}\alpha\Big)\Big)
+
\int_N \beta\wedge d\Big(\mapint_{\pi_\Nbun}\alpha\Big)
\\&=
(-1)^{n-p}\int_{\partial N} \beta\wedge\Big(\mapint_{\pi_\Nbun}\alpha\Big)
+
\int_N \beta\wedge d\Big(\mapint_{\pi_\Nbun}\alpha\Big)
\\&=
\int_N \beta\wedge d\Big(\mapint_{\pi_\Nbun}\alpha\Big)
\end{align*}
\end{proof}

\begin{lemma}
Proof of
\begin{align}
\mapint_{\PiX} \star_X \alpha = \star \mapint_{\PiX} \alpha
\label{Hodge_int_commute}
\end{align}
\end{lemma}
\begin{proof}
The only non-trivial $\alpha\in\Gamma\Lambda(\MX\times\MY)$ in
(\ref{Hodge_int_commute}) can be written $\alpha=\alpha_{I}dx^{I}\wedge
dy^{0123}$. Then
\begin{align*}
\mapint_{\PiX} \star_X \big(\alpha_{I}dx^{I}\wedge dy^{0123}\big)
&=
\mapint_{\PiX} \alpha_{I}(\star dx^{I})\wedge dy^{0123}
=
\star dx^{I}\int_{\MX} \alpha_{I} dy^{0123}
=
\star  \mapint_{\PiX} \alpha_{I} dx^{I}\wedge dy^{0123}
\end{align*}
\end{proof}

\begin{lemma}
\label{lm_Causal}
$\Pi$ is causal on $\MXp$ if and only if

\begin{align}
\begin{minipage}{0.9\textwidth}
\begin{itemize}
\item
$Z$ is causal on $\MXp$,
\item
$(d_Y\chi)|_{(x,y)}=0$
for all $(x,y)\in\MXp\times\MYp$ such that $y \notin J^-(x)$ and
\item
$\iota_{\SigmaMY}^\star(\chi)\,\big|\,_{(x,y)}=0$ for all
$(x,y)\in\MXp\times\SigmaMY$ such that $y \notin J^-(x)$, where
$\iota_{\SigmaMY}:\MXp\times\SigmaMY\hookrightarrow\MXp\times\MYp$
is the natural embedding.
\end{itemize}
\end{minipage}
\label{Proofs_causal_cond}
\end{align}

\end{lemma}
\begin{proof}
If $\hat{\iota}_{\SigmaMY}:\SigmaMY\hookrightarrow\MYp$ is the
natural embedding then $\IC{x}{ab}\iota_{\SigmaMY}^\star \chi|_{(x,y)}
=\hat{\iota}_{\SigmaMY}^\star\IC{x}{ab}\chi|_{(x,y)}$, and
\begin{equation}
\begin{aligned}
\lefteqn{\int_{y\in\MYp} \IC{x}{ab} \big(\chi| \wedge \PiY^\star (dA|_y)\big)}
\quad&
\\&=
\int_{y\in\MYp} \IC{x}{ab} \big(\chi \wedge d_Y (\PiY^\star A)\big)|_{(x,y)}
\\&=
\int_{y\in\MYp} \IC{x}{ab} d_Y \big(\chi \wedge \PiY^\star
A|_y\big)\big|_{(x,y)}
-
\int_{y\in\MYp} \IC{x}{ab} \big(d_Y \chi \wedge \PiY^\star
A|_y\big)\big|_{(x,y)}
\\&=
\int_{y\in\MYp} d_Y \big(\IC{x}{ab} \chi \wedge \PiY^\star
A|_y\big)\big|_{(x,y)}
-
\int_{y\in\MYp} \IC{x}{ab} \big(d_Y \chi \wedge \PiY^\star
A|_y\big)\big|_{(x,y)}
\\&=
\int_{y\in\SigmaMY} \hat{\iota}_{\SigmaMY}^\star \IC{x}{ab} \big(\chi \wedge \PiY^\star
A|_y\big)\big|_{(x,y)}
-
\int_{y\in\MYp} \IC{x}{ab} \big(d_Y \chi \wedge \PiY^\star
A|_y\big)\big|_{(x,y)}
\\&=
\int_{y\in\SigmaMY\backslash J^-(x)} \IC{x}{ab} \iota_{\SigmaMY}^\star \big(\chi \wedge \PiY^\star
A|_y\big)\big|_{(x,y)}
+
\int_{y\in\SigmaMY\inter J^-(x)} \IC{x}{ab} \iota_{\SigmaMY}^\star \big(\chi \wedge \PiY^\star
A|_y\big)\big|_{(x,y)}
\\&\quad-
\int_{y\in\MYp\backslash J^-(x)} \IC{x}{ab} \big(d_Y \chi \wedge \PiY^\star
A|_y\big)\big|_{(x,y)}
-
\int_{y\in\MYp\inter J^-(x)} \IC{x}{ab} \big(d_Y \chi \wedge \PiY^\star
A|_y\big)\big|_{(x,y)}
\end{aligned}
\label{Proofs_int_chi}
\end{equation}

First one argues that (\ref{Proofs_causal_cond}) implies that $\Pi$ is
causal on $\MXp$.
Given $x\in\MXp$ and
$F_1,F_2\in\Gamma\Lambda^2\MYp$ such that $F_1|_y = F_2|_y=0$ for
$y\in J^-(y)$, set $F=F_1-F_2$ so that $F=0$ on $J^-(x)$. Since $\MYp$
is topologically trivial $F$ is exact, $F=d\hat{A}$, and hence
$d\hat{A}=0$ on $J^-(x)$. Then since $J^-(x)$ is topologically trivial
there exists $f\in\Gamma\Lambda^0\MYp$ such that $\hat{A}=d f$ on
$J^-(x)$. Thus one can choose a gauge $A=\hat{A}-d f$ so that $A=0$ on
$J^-(x)$. Given $\zetabd$ such that $\zetabd|_y=0$ for $y\in
J^-(x)\inter\SigmaMY$ then $Z[\zetabd]|_x=0$ since $Z$ is causal. Thus
from (\ref{Proofs_int_chi})
\begin{align*}
\Pi[F,\zetabd]_{ab}(x)
&=
\int_{y\in\MYp} \IC{x}{ab} \big(\chi| \wedge \PiY^\star (dA|_y)\big)
\\&=
\int_{y\in\SigmaMY\backslash J^-(x)} \IC{x}{ab} \iota_{\SigmaMY}^\star \big(\chi \wedge \PiY^\star
A|_y\big)\big|_{(x,y)}
+
\int_{y\in\SigmaMY\inter J^-(x)} \IC{x}{ab} \iota_{\SigmaMY}^\star \big(\chi \wedge \PiY^\star
A|_y\big)\big|_{(x,y)}
\\&\quad-
\int_{y\in\MYp\backslash J^-(x)} \IC{x}{ab} \big(d_Y \chi \wedge \PiY^\star
A|_y\big)\big|_{(x,y)}
-
\int_{y\in\MYp\inter J^-(x)} \IC{x}{ab} \big(d_Y \chi \wedge \PiY^\star
A|_y\big)\big|_{(x,y)}
=
0
\end{align*}
since $\iota_{\SigmaMY}^\star \chi|_{(x,y)}=0$ for
$y\in\SigmaMY\backslash J^-(x)$, $A|_y=0$ for $y\in J^-(x)$ and
$d_Y\chi=0$ for $y\in\MYp\backslash J^-(x)$.

Conversely if $\Pi$ is causal on $\MXp$ then setting $F=0$ in
(\ref{Disp_Pi}) shows that $Z$ must be causal on $\MXp$. Then setting
$\zetabd=0$ then for all $A$ such that
$A=0$ on $J^-(x)$ (\ref{Proofs_int_chi}) yields
\begin{equation}
\begin{aligned}
0 &= \Pi[F,\zetabd]_{ab}(x)
\\&=
\int_{y\in\SigmaMY\backslash J^-(x)} \IC{x}{ab} \iota_{\SigmaMY}^\star \big(\chi_{(x,y)} \wedge \PiY^\star
A|_y\big)\big|_{(x,y)}
-
\int_{y\in\MYp\backslash J^-(x)} \IC{x}{ab} \big(d_Y \chi_{(x,y)} \wedge \PiY^\star
A|_y\big)\big|_{(x,y)}
\end{aligned}
\label{Proofs_Causal_0}
\end{equation}
The 4-dimensional domain $\MYp\backslash J^-(x)$ denotes points
outside the backward lightcone of $x$, while the $3$-dimensional
domain $\SigmaMY\backslash J^-(x)$ denotes the points on $\SigmaMY$
that are not causally connected to $x$. Choosing such an $A$ to have
support about a small neighbourhood of $y\in\MYp\backslash J^-(x)
\backslash \SigmaMY$ results in the first term of
(\ref{Proofs_Causal_0}) being zero and thus $(d_Y \chi)|_{(x,y)}=0$.
Likewise setting $A$ to have support about a small neighbourhood of
$y\in\SigmaMY\backslash J^-(x)$ implies $\iota_{\SigmaMY}^\star
(\chi)|_{(x,y)}=0$.

\end{proof}

\noindent One can now prove {\bf lemma \ref{lm_int_chac_curves}} in
section \ref{ch_SusKernel}.

\begin{proof}[Proof of lemma \ref{lm_int_chac_curves}]
Given $\sigma\in N$, with $V$ non vanishing there exists a coordinate
system $(\sigma^1,\ldots,\sigma^n)$ on $N$ adapted to $V$ so that
$V=\pfrac{}{\sigma^1}$ and the image of the curve
$\gamma_\sigma:[\tau_0(\sigma),0]\to N$ is contained in the coordinate
patch.  Write $\beta=\beta_I d\sigma^I$ then since $i_V\beta=0$ the
sum is over $I\in\Set{2,\ldots,n}$. With $\sigma^1$ distinguished
write $\beta_I(\sigma)=\beta_I(\sigma^1,\Vsig)$ where
$\Vsig=(\sigma^2,\ldots,\sigma^n)$. Also since $i_V\beta=0$,
$\beta|_{(\sigma^1,\Vsig)}=\beta_I(\sigma^1,\Vsig) d\Vsig^I$.
Likewise since $i_V\zeta=0$ one has
$\zeta|_{\sigma_0}=\zeta_{I}(\sigma_0) d\Vsig^I$.

Solving for the integral curves of $V$ gives
$\phi_N(\tau,\sigma^1,\Vsig)=(\tau+\sigma^1,\Vsig)$
\begin{align*}
\phi_N^\star(\beta)|_{(\tau,\sigma^1,\Vsig)}=\beta_I(\tau+\sigma^1,\Vsig)
d\Vsig^I
\end{align*}
and one may write $\tau_0(\sigma^1,\Vsig)=\tau_0(\Vsig)-\sigma^1$, giving
\begin{align*}
\varphi_{N}(\zeta)|_{(\sigma^1,\Vsig)} =
\phi^\star_{N\tau_0(\sigma^1,\Vsig)}(\zeta|_{\tau_0(\Vsig)}) =
\zeta_{I}\big(\tau_0(\Vsig),\Vsig\big) d\Vsig^I
\end{align*}
Thus
\begin{align*}
\xi|_{(\sigma^1,\Vsig)}
&=
\mapint_{\varpi_N}\phi_N^\star(\beta) \wedge d\tau
+
\varphi_{N}(\zeta)
\\&=
\Big(\int_{\tau=\tau_0(\Vsig)-\sigma^1}^0
\beta_I(\sigma^1+\tau,\Vsig) d\tau
+
\zeta_{I}\big(\tau_0(\Vsig),\Vsig\big)
\Big) d\Vsig^I
\end{align*}
Hence $i_V\xi=0$ and one may write
$\xi|_{(\sigma^1,\Vsig)}=\xi_I(\sigma^1,\Vsig) d\Vsig^I$. Now
\begin{align*}
\xi_I(\sigma^1,\Vsig)
&=
\int_{\tau=\tau_0(\Vsig)-\sigma^1}^0
\beta_I(\sigma^1+\tau,\Vsig) d\tau
+
\zeta_{I}\big(\tau_0(\Vsig),\Vsig\big)
\\&=
\int_{\tau=\tau_0(\Vsig)}^{\sigma^1}
\beta_I(\tau',\Vsig) d\tau'
+
\zeta_{I}\big(\tau_0(\Vsig),\Vsig\big)
\end{align*}
where $\tau'=\tau+\sigma^1$ and
\begin{align*}
i_V d \xi|_{(\sigma^1,\Vsig)}
&=
i_{\pfrac{}{\sigma^1}} d(\xi_I(\sigma^1,\Vsig) d\Vsig^I)
=
i_{\pfrac{}{\sigma^1}} \big(d\xi_I(\sigma^1,\Vsig) \wedge d\Vsig^I\big)
=
\pfrac{\xi_I(\sigma^1,\Vsig)}{\sigma^1} d\Vsig^I
\\&=
\pfrac{}{\sigma^1}
\Big(\int_{\tau=\tau_0(\Vsig)}^{\sigma^1}
\beta_I(\tau',\Vsig) d\tau'
+
\zeta_{I}\big(\tau_0(\Vsig),\Vsig\big)
\Big) d\Vsig^I
=
\beta_I(\sigma^1,\Vsig)d\Vsig^I
=
\beta|_{(\sigma^1,\Vsig)}
\end{align*}
Since $\sigma^1=0$ on $\SigmaN$
\begin{align*}
\xi|_{(0,\Vsig)}
=
\xi_I(\tau_0(0,\Vsig),\Vsig) d\Vsig^I
=
\xi_I(\tau_0(\Vsig),\Vsig) d\Vsig^I
=
\zeta_{I}(\tau_0(\Vsig),\Vsig) d\Vsig^I
=
\zeta|_{(0,\Vsig)}
\end{align*}
i.e. $\xi|_{\SigmaN}=\zeta$.

\end{proof}

\begin{lemma}
\label{lm_chi_nocoords}
Proof that
\textup{(\ref{Plasma_chi_sum_species},\ref{Plasma_res_chi})} implies
\textup{(\ref{SusKernel_Pi1_eq})} and that
\textup{(\ref{Plasma_chi_sum_species},\ref{Plasma_res_chi_imp})} implies
\textup{(\ref{SusKernel_Pi1_eq})}.
\end{lemma}
\begin{proof}
First (\ref{Plasma_res_chi}) is equivalent to
(\ref{Plasma_res_chi_imp}) since given
$\gamma\in\Gamma\Lambda^2\MYp$ one has
$\dual{i_{(y,u)}\gamma}=u^a\gamma_{ab} g^{bc} \pfrac{}{y^c}$ and hence
$\What^\speciesa(\gamma)=
\frac{q^{\speciesa}}{m^\speciesa}\Vert_{(y,u)}(\dual{i_{(y,u)}\gamma})
=\frac{q^{\speciesa}}{m^\speciesa}u^a\gamma_{ab} g^{b\nu} \pfrac{}{u^\nu}$.
From (\ref{Plasma_res_chi}) it follows that
\begin{align*}
\chi^\speciesa\wedge\PiY^\star \gamma
&=
\tfrac1{2} \frac{q^{\speciesa\,2}}{m^\speciesa}\hashx
dy^{cd}
\wedge\IC{y}{abcd}
\Psi^{\speciesa\star} \Big( d\tau \wedge
\varpi_Y^{\speciesa\star}
\big(g^{\nu a} u^b \IC{u}{\nu} \theta^\speciesa_0\big)\Big)\wedge
\PiY^\star \gamma
\\&=
\frac{q^{\speciesa\,2}}{m^\speciesa}\hashx
S\IC{y}{ab}
\Psi^{\speciesa\star} \Big( d\tau \wedge
\varpi_Y^{\speciesa\star}
\big(g^{\nu a} u^b \IC{u}{\nu} \theta^\speciesa_0\big)\Big)\wedge
\PiY^\star \gamma
\\&=
-\frac{q^{\speciesa\,2}}{m^\speciesa}\hashx
S\Psi^{\speciesa\star} \Big( d\tau \wedge
\varpi_Y^{\speciesa\star}
\big(g^{\nu a} u^b \IC{u}{\nu} \theta^\speciesa_0\big)\Big)\wedge
\IC{y}{ab}\PiY^\star \gamma
\\&=
-\frac{q^{\speciesa\,2}}{m^\speciesa}\hashx
S\Psi^{\speciesa\star} \Big( d\tau \wedge
\varpi_Y^{\speciesa\star}
\big(\gamma_{ab}g^{\nu a} u^b \IC{u}{\nu} \theta^\speciesa_0\big)\Big)
\\&=
-q^{\speciesa}\hashx
S\Psi^{\speciesa\star} \Big( d\tau \wedge
\varpi_Y^{\speciesa\star}
\big(i_{\What^\speciesa(\gamma)}\theta^\speciesa_0\big)\Big)
\end{align*}
i.e. (\ref{Plasma_res_chi_imp}). That
(\ref{Plasma_res_chi_imp}) implies (\ref{Plasma_res_chi}) follows
since the above argument is true for all $\gamma$.

To prove (\ref{SusKernel_Pi1_eq}) note that the domains
$\Nbun^\speciesa_X$ and $\Nbun^\speciesa_Y$ are related via the
diffeomorphism
\begin{align}
\Upsilon^\speciesa:\Nbun^\speciesa_Y\to\Nbun^\speciesa_X\,,\qquad
\Upsilon^\speciesa(\tau,y,u)=\big(-\tau,\Cd^\speciesa_{(y,u)}(\tau)\big)
\label{Proofs_def_Upsilon}
\end{align}
Thus $\Upsilon^{\speciesa\star}(d\tau)=-d\tau$ and setting
$(x,v)=\Cd^\speciesa_{(y,u)}(\tau)$ with $\tau>0$ yields
\begin{align*}
\phi^{\speciesa}\big(\Upsilon^\speciesa(\tau,y,u)\big)
=
\phi^{\speciesa}\big(-\tau,\Cd^\speciesa_{(y,u)}(\tau)\big)
=
\phi^{\speciesa}(-\tau,x,v)
=
(y,u)
=
\ppY^\speciesa(\tau,y,u)
\end{align*}
so that $\ppY^\speciesa=\phi^{\speciesa}\circ\Upsilon^\speciesa$ and
thus
$\ppY^{\speciesa\star}=\Upsilon^{\speciesa\star}\circ\phi^{\speciesa\star}$.
Now
\begin{align*}
\Upsilon^{\speciesa\star}\Big(d\tau \wedge
\phi^{\speciesa\star}
\big(i_{\What^\speciesa(F_1)}\theta^\speciesa_0\big)\Big)
=
\Upsilon^{\speciesa\star}(d\tau) \wedge
\Upsilon^{\speciesa\star}\phi^{\speciesa\star}
\big(i_{\What^\speciesa(F_1)}\theta^\speciesa_0\big)
=
-d\tau \wedge
\varpi_Y^{\speciesa\star}
\big(i_{\What^\speciesa(F_1)}\theta^\speciesa_0\big)
\end{align*}
hence
\begin{align}
\chi^\speciesa\wedge\PiY^\star F_1
&=
q^{\speciesa}\hashx
S\Psi^{\speciesa\star}
\Upsilon^{\speciesa\star}\Big(d\tau \wedge
\phi^{\speciesa\star}
\big(i_{\What^\speciesa(F_1)}\theta^\speciesa_0\big)\Big)
\label{Proofs_chispecies_F1}
\end{align}

From (\ref{Plasma_def_varphi})
\begin{align*}
\PiX\big(\Phi^{\speciesa}(\tau,y,u)\big)=
\PiX\big(C^\speciesa_{(y,u)}(\tau),y\big)=
C^\speciesa_{(y,u)}(\tau)
\end{align*}
and from (\ref{Proofs_def_Upsilon})
\begin{align*}
\piX\big(\ppX^\speciesa\big(\Upsilon^\speciesa(\tau,y,u)\big)\big)
=
\piX\big(\ppX^\speciesa\big(-\tau,\Cd^\speciesa_{(y,u)}(\tau)\big)\big)
=
\piX\big(\Cd^\speciesa_{(y,u)}(\tau)\big)
=
C^\speciesa_{(y,u)}(\tau)
\end{align*}
Hence $\PiX\circ\Phi^{\speciesa}=
\piX\circ\ppX^\speciesa\circ\Upsilon^\speciesa$ and so
\begin{align}
\Phi^{\speciesa\star}\circ\PiX^\star
=
\Upsilon^{\speciesa\star}\circ\ppX^{\speciesa\star}\circ\pi_X^\star
\label{Proofs_maps_stars}
\end{align}

From the definition of $S$ one has
\begin{align}
\mapint_{\PiX} S \gamma = \mapint_{\PiX} \gamma
\label{Proofs_S}
\end{align}
for any $\gamma\in\Gamma\Lambda^8(\MX\times\MY)$.

Since $\Psi^{\speciesa}:\DomPsi\to\DomPsi'$ is a diffeomorphism then
\begin{align}
\int_\DomPsi\Psi^{\speciesa\star}\gamma = \int_{\DomPsi'}\gamma
\label{Proofs_int_DomPsi}
\end{align}
for any $\gamma\in\Gamma\Lambda^8(\DomPsi')$. Likewise since
$\Upsilon^\speciesa:\Nbun^\speciesa_Y\to\Nbun^\speciesa_X$ is a
diffeomorphism
\begin{align}
\int_{\Nbun^\speciesa_Y}\Upsilon^{\speciesa\star}\gamma =
\int_{\Nbun^\speciesa_X}\gamma
\label{Proofs_int_Upsilon}
\end{align}
for any $\gamma\in\Gamma\Lambda^8(\Nbun^\speciesa_X)$.

For convenience set $\alpha^\speciesa=d\tau \wedge
\phi^{\speciesa\star}
\big(i_{\What^\speciesa(F_1)}\theta^\speciesa_0\big)
\in\Gamma\Lambda^5\Nbun^\speciesa_X$.
For fixed $x$ assume that $F_1$ has support in $\DomPsi_x$. Then one
can choose $\beta\in\Gamma\Lambda^2\MX$ so that
$\PiX^\star\beta\wedge\PiY^\star F_1$ has support inside
$\DomPsi$. Thus from (\ref{Proofs_chispecies_F1})
\begin{align}
\supp\big(\PiX^\star(\star\beta)\wedge
\Psi^{\speciesa\star}
\Upsilon^{\speciesa\star}
\alpha^\speciesa\big)
=
\supp\big(\PiX^\star\beta\wedge\chi^\speciesa\wedge\PiY^\star F_1\big)
\subset\DomPsi
\label{Proofs_supp_betaalpha}
\end{align}
Now
\begin{align*}
\int_{\MX}\beta\wedge\mapint_{\PiX}
\chi^\speciesa\wedge\PiY^\star F_1
&=
\int_{\MX}\beta\wedge\mapint_{\PiX}
q^{\speciesa}\hashx S\Psi^{\speciesa\star}
\Upsilon^{\speciesa\star}
\alpha^\speciesa
&&
\text{from (\ref{Proofs_chispecies_F1})}
\\&=
q^{\speciesa}\int_{\MX}\beta\wedge\star\mapint_{\PiX}
S\Psi^{\speciesa\star}
\Upsilon^{\speciesa\star}
\alpha^\speciesa
&&
\text{from (\ref{Hodge_int_commute})}
\displaybreak[0]\\&=
q^{\speciesa}\int_{\MX}\beta\wedge\star\mapint_{\PiX}
\Psi^{\speciesa\star}
\Upsilon^{\speciesa\star}
\alpha^\speciesa
&&
\text{from (\ref{Proofs_S})}
\displaybreak[0]\\&=
-q^{\speciesa}\int_{\MX}(\star\beta)\wedge\mapint_{\PiX}
\Psi^{\speciesa\star}
\Upsilon^{\speciesa\star}
\alpha^\speciesa
\displaybreak[0]\\&=
-q^\speciesa\int_{\MX\times\MY}\PiX^\star(\star\beta)\wedge
\Psi^{\speciesa\star}
\Upsilon^{\speciesa\star}
\alpha^\speciesa
&&
\text{from (\ref{Notation_mapint})}
\displaybreak[0]\\&=
-q^\speciesa\int_{\DomPsi}\PiX^\star(\star\beta)\wedge
\Psi^{\speciesa\star}
\Upsilon^{\speciesa\star}
\alpha^\speciesa
&&
\text{from (\ref{Proofs_supp_betaalpha})}
\displaybreak[0]\\&=
-q^\speciesa\int_{\DomPsi}\Psi^{\speciesa\star}
\Big(\Phi^{\speciesa\star}\PiX^\star(\star\beta)\wedge
\Upsilon^{\speciesa\star}
\alpha^\speciesa\Big)
&&
\text{from (\ref{Plasma_def_psi})}
\displaybreak[0]\\&=
-q^\speciesa\int_{\DomPsi'}
\Phi^{\speciesa\star}\PiX^\star(\star\beta)\wedge
\Upsilon^{\speciesa\star}
\alpha^\speciesa
&&
\text{from (\ref{Proofs_int_DomPsi})}
\displaybreak[0]\\&=
-q^\speciesa\int_{\Nbun^\speciesa_Y}
\Phi^{\speciesa\star}\PiX^\star(\star\beta)\wedge
\Upsilon^{\speciesa\star}
\alpha^\speciesa
&&
\text{since $\DomPsi'\subset\Nbun_Y^\speciesa$}
\displaybreak[0]\\&=
-q^\speciesa\int_{\Nbun^\speciesa_Y}
\Upsilon^{\speciesa\star}\ppX^{\speciesa\star}\pi_X^\star
(\star\beta)\wedge
\Upsilon^{\speciesa\star}
\alpha^\speciesa
&&
\text{from (\ref{Proofs_maps_stars})}
\displaybreak[0]\\&=
-q^\speciesa\int_{\Nbun^\speciesa_X}
\ppX^{\speciesa\star}\pi_X^\star
(\star\beta)\wedge
\alpha^\speciesa
&&
\text{from (\ref{Proofs_int_Upsilon})}
\displaybreak[0]\\&=
-q^\speciesa\int_{\EbunX}
\pi_X^\star
(\star\beta)\wedge
\mapint_{\ppX^\speciesa}\alpha^\speciesa
&&
\text{from (\ref{Notation_mapint})}
\displaybreak[0]\\&=
-q^\speciesa\int_{\MX}
(\star\beta)\wedge
\mapint_{\piX}\mapint_{\ppX^\speciesa}\alpha^\speciesa
&&
\text{from (\ref{Notation_mapint})}
\\&=
q^\speciesa\int_{\MX}
\beta\wedge
\star\mapint_{\piX}\mapint_{\ppX^\speciesa}\alpha^\speciesa
\end{align*}
Summing over $\speciesabig$ gives
\begin{align*}
\int_{\MX}\beta\wedge\mapint_{\PiX}
\chi\wedge\PiY^\star F_1
&=
\sum_{\speciesa} q^\speciesa\int_{\MX}
\beta\wedge
\star\mapint_{\piX}\mapint_{\ppX^\speciesa}\alpha^\speciesa
\end{align*}
Since this is true for all $\beta$ with support in a neighbourhood of
$x$ then (\ref{SusKernel_Pi1_eq}) holds at $x$.
\end{proof}

\begin{lemma}
\label{lm_chi_coords}
The derivation of \textup{(\ref{Plasma_coords})} from
\textup{(\ref{Plasma_res_chi})}.
\end{lemma}
\begin{proof}
The derivation of (\ref{Plasma_coords}) from (\ref{Plasma_res_chi})
follows by first writing the Liouville vector field
(\ref{Plasma_MV_W0}) as
\begin{align*}
W^\speciesa_0=u^a\pfrac{}{y^a}+H^\inu \pfrac{}{u^\inu}
\qquadtext{where}
H^\inu=-\Gamma^\inu{}_{ef}u^e u^f +
\frac{q^\speciesa}{m^\speciesa}F_{0ef} g^{\inu e} u^f
\end{align*}
Then setting $f^\speciesa(y,u)=f^\speciesa_0(y,u)+f^\speciesa_1(y,u)$
it follows from (\ref{Plasma_f}) that
\begin{align*}
\theta^\speciesa_0
&=
i_{W^\speciesa_0}(f_0^\speciesa\Omega)
=
f_0^\speciesa i_{W^\speciesa_0}\Big(\frac{\detg}{u_0} dy^{0123}\wedge du^{123}\Big)
\\&=
f_0^\speciesa \frac{\detg}{u_0} \Big( u^c Y_c\wedge du^{123} +\tfrac12 H^\imu
\epsilon_{\imu\inu\isigma} Y\wedge du^{\inu \isigma}\Big)
\end{align*}
where $Y_a=i_{\pfrac{}{y^a}} dy^{0123}$ and $Y=dy^{0123}$. Consequently
\begin{align*}
g^{\mu a} u^b \IC{u}{\mu}\theta^\speciesa_0=
f_0^\speciesa \frac{\detg}{u_0} g^{\imu a} u^b
\Big(-\tfrac12 u^c \epsilon_{\imu\inu\isigma} Y_c \wedge du^{\inu \isigma} -
H^\inu\epsilon_{\imu\inu\isigma} Y\wedge du^{\isigma}\Big)
\end{align*}
and
\begin{align*}
-d\tau\wedge g^{\mu a} u^b \IC{u}{\mu}\theta^\speciesa_0=
f_0^\speciesa \frac{\detg}{u_0} g^{\imu a} u^b\epsilon_{\imu\inu\isigma}
\Big(\frac{u^c}{2}  d\tau\wedge Y_c \wedge du^{\inu \isigma} +
H^\inu d\tau \wedge Y\wedge du^{\isigma}\Big)
\end{align*}
Under the maps $\ppY^\speciesa$ and $\psihat^{\speciesa\star}$ one has
\begin{align*}
\ppY^{\speciesa\star} (dy^a)=dy^a
\,,\qquad
\ppY^{\speciesa\star} (du^\mu)=du^\mu
\end{align*}
and
\begin{gather*}
\psihat^{\speciesa\star}(dy^a)=dy^a
\,,\quad
\psihat^{\speciesa\star}(du^\imu)=\pfrac{u^\imu}{x^a} dx^a + \pfrac{u^\imu}{y^a} dy^a
\,,\quad
\psihat^{\speciesa\star}(d\tau)=\pfrac{\tau}{x^a} dx^a + \pfrac{\tau}{y^a} dy^a
\end{gather*}
So using the projector $S$ given in (\ref{Plasma_def_S}) yields
\begin{align*}
-S \psihat^{\speciesa\star} \big( d\tau\wedge g^{\nu a} u^b \IC{u}{\nu}\theta^\speciesa_0 \big)=
f_0^\speciesa \frac{\detg}{u_0} g^{\imu a} u^b
\epsilon_{\imu\inu\isigma}\Big(
\frac{u^c}{2}\pfrac{\tau}{y^c}\pfrac{u^\inu}{x^d}\pfrac{u^\isigma}{x^e}
-\frac{u^c}{2}\pfrac{\tau}{x^d}\pfrac{u^\inu}{y^c}\pfrac{u^\isigma}{x^e}
&\\
+\frac{u^c}{2}\pfrac{\tau}{x^d}\pfrac{u^\inu}{x^e}\pfrac{u^\isigma}{y^c}
+H^\inu\pfrac{\tau}{x^d}\pfrac{u^\isigma}{x^e}
\Big) Y\wedge dx^{de} &
\end{align*}
Hence from (\ref{Plasma_res_chi})
\begin{align*}
\chi^\speciesa
&=
-\frac{q^{\speciesa2}}{m^\speciesa}\hashx \Big(
\IC{y}{ab}
S \psihat^{\speciesa\star} \Big( d\tau \wedge
\ppY^\star (g^{\nu a} u^b \IC{u}{\nu} \theta^\speciesa_0)\Big)\Big)
\\
&=
\frac{q^{\speciesa2}}{m^\speciesa}\hashx
\IC{y}{ab}\bigg(
f_0^\speciesa \frac{\detg}{u_0} g^{\imu a} u^b
\epsilon_{\imu\inu\isigma}\Big(
\frac{u^c}{2}\pfrac{\tau}{y^c}\pfrac{u^\inu}{x^d}\pfrac{u^\isigma}{x^e}
-\frac{u^c}{2}\pfrac{\tau}{x^d}\pfrac{u^\inu}{y^c}\pfrac{u^\isigma}{x^e}
\\
&\hspace{18em}
+\frac{u^c}{2}\pfrac{\tau}{x^d}\pfrac{u^\inu}{x^e}\pfrac{u^\isigma}{y^c}
+H^\inu\pfrac{\tau}{x^d}\pfrac{u^\isigma}{x^e}
\Big) Y\wedge dx^{de}\bigg)
\\&=
-\hashx
\bigg(
f_0^\speciesa \frac{\detg}{u_0} g^{\imu a} u^b
\epsilon_{\imu\inu\isigma}\epsilon_{abfg}\Big(
\frac{u^c}{2}\pfrac{\tau}{y^c}\pfrac{u^\inu}{x^d}\pfrac{u^\isigma}{x^e}
-\frac{u^c}{2}\pfrac{\tau}{x^d}\pfrac{u^\inu}{y^c}\pfrac{u^\isigma}{x^e}
\\
&\hspace{18em}
+\frac{u^c}{2}\pfrac{\tau}{x^d}\pfrac{u^\inu}{x^e}\pfrac{u^\isigma}{y^c}
+H^\inu\pfrac{\tau}{x^d}\pfrac{u^\isigma}{x^e}
\Big) dx^{de}\wedge dy^{fg}\bigg)
\\&=
\frac{q^{\speciesa2}}{m^\speciesa} f_0^\speciesa \frac{\detg^{3/2}}{2 u_0}  g^{\imu b} u^a
\epsilon_{\imu\inu\isigma}\epsilon_{abfg}\epsilon^{dehi}\Big(
\frac{u^c}{2}\pfrac{\tau}{y^c}\pfrac{u^\inu}{x^d}\pfrac{u^\isigma}{x^e}
-\frac{u^c}{2}\pfrac{\tau}{x^d}\pfrac{u^\inu}{y^c}\pfrac{u^\isigma}{x^e}
\\
&\hspace{18em}
+\frac{u^c}{2}\pfrac{\tau}{x^d}\pfrac{u^\inu}{x^e}\pfrac{u^\isigma}{y^c}
+H^\inu\pfrac{\tau}{x^d}\pfrac{u^\isigma}{x^e}
\Big) dx_{hi}\wedge dy^{fg}
\end{align*}
\end{proof}

\end{document}